\documentclass[twoside,11pt]{article}

\usepackage{jmlr2e}
\usepackage{amsmath}
\usepackage{bm}
\usepackage{algorithm}
\usepackage{algorithmic}

\newcommand{\ttau}{\tilde{\tau}}
\newcommand{\tsigma}{\tilde{\sigma}}
\newcommand{\tilbw}{\tilde{\bw}}
\newcommand{\bA}{\mathbf{A}}
\newcommand{\bQ}{\mathbf{Q}}
\newcommand{\bS}{\mathbf{S}}
\newcommand{\bP}{\mathbf{P}}
\newcommand{\bq}{\mathbf{q}}
\newcommand{\bu}{\mathbf{u}}

\newcommand{\bH}{\mathbf{H}}
\newcommand{\bB}{\mathbf{B}}
\newcommand{\bM}{\mathbf{M}}
\newcommand{\ba}{\mathbf{a}}
\newcommand{\bb}{\mathbf{b}}
\newcommand{\bW}{\mathbf{W}}
\newcommand{\bx}{\mathbf{x}}
\newcommand{\bX}{\mathbf{X}}
\newcommand{\by}{\mathbf{y}}
\newcommand{\bY}{\mathbf{Y}}
\newcommand{\barq}{\bar{q}}
\newcommand{\bs}{\mathbf{s}}

\newcommand{\bI}{\mathbf{I}}
\newcommand{\bw}{\mathbf{w}}
\newcommand{\bTheta}{\bm{\Theta}}
\newcommand{\bZ}{\mathbf{Z}}

\newcommand{\tilbZ}{\tilde{\bZ}}
\newcommand{\hatbW}{\hat{\bW}}
\newcommand{\tilby}{\tilde{\by}}
\newcommand{\bE}{\mathbf{E}}
\newcommand{\tilf}{\tilde{f}}
\newcommand{\hatbx}{\hat{\bx}}
\newcommand{\boldm}{\mathbf{m}}
\newcommand{\tilw}{\tilde{w}}
\newcommand{\hatZ}{\hat{Z}}
\newcommand{\tilbA}{\tilde{\bA}}
\newcommand{\hatx}{\hat{x}}
\newcommand{\bone}{\mathbf{1}}

\newcommand{\tilT}{\tilde{T}}
\newcommand{\calY}{\mathcal{Y}}
\newcommand{\tilbQ}{\tilde{\bQ}}
\newcommand{\tilbM}{\tilde{\bM}}
\newcommand{\tSigma}{\tilde{\Sigma}}

\newcommand{\calA}{\mathcal{A}}

\newcommand{\calF}{\mathcal{F}}

\newcommand{\calN}{\mathcal{N}}
\newcommand{\calH}{\mathcal{H}}
\newcommand{\calZ}{\mathcal{Z}}
\newcommand{\calB}{\mathcal{B}}

\newcommand{\tilP}{\tilde{P}}
\newcommand{\tilW}{\tilde{W}}
\newcommand{\tilE}{\tilde{E}}
\newcommand{\tili}{\tilde{i}}
\newcommand{\eps}{\varepsilon}
\newcommand{\bz}{\mathbf{z}}
\newcommand{\bh}{\mathbf{h}}

\newcommand{\hatS}{\hat{S}}
\newcommand{\hatbS}{\hat{\bS}}
\newcommand{\hatW}{\hat{W}}
\newcommand{\hatbZ}{\hat{\bZ}}

\newcommand{\tilY}{\tilde{Y}}

\newcommand{\calG}{\mathcal{G}} 
\newcommand{\tilbW}{\tilde{\bW}} 
\newcommand{\btheta}{\bm{\theta}}
\newcommand{\barx}{\bar{x}}
\newcommand{\barX}{\bar{X}}

\newcommand{\bbR}{\mathbb{R}} 
\newcommand{\bbZ}{\mathbb{Z}} 
\newcommand{\bbE}{\mathbb{E}} 
\newcommand{\bbP}{\mathbb{P}} 

\DeclareMathOperator*{\argmin}{arg\,min}
\newcommand{\nn}{\nonumber} 
\DeclareMathOperator{\var}{\mathsf{Var}}

\usepackage{lastpage}
\jmlrheading{24}{2023}{1-\pageref{LastPage}}{5/21; Revised 2/23}{2/23}{21-0543}{Lan V. Truong} 
\ShortHeadings{Sparse linear regression with sublinear sparsity}{Truong}
\firstpageno{1}

\begin{document}

\title{Fundamental limits and algorithms for sparse linear regression with sublinear sparsity}

\author{\name Lan V. Truong \email lt407@cam.ac.uk \\
       \addr Department of Engineering\\
       University of Cambridge\\
       Cambridge, CB2 1PZ, UK}
      
\editor{Pierre Alquier}

\maketitle

\begin{abstract}
We establish exact asymptotic expressions for the normalized mutual information and minimum mean-square-error (MMSE) of  sparse linear regression in the sub-linear sparsity regime. Our result is achieved by a generalization of the adaptive interpolation method in Bayesian inference for linear regimes to sub-linear ones. A modification of the well-known approximate message passing algorithm to approach the MMSE fundamental limit is also proposed, and its state evolution is rigorously analysed.  Our results show that the traditional linear assumption between the signal dimension and number of observations in the replica and adaptive interpolation methods is not necessary for sparse signals. They also show how to modify the existing well-known AMP algorithms for linear regimes to sub-linear ones.
\end{abstract}


\section{Introduction} 
The estimation of a signal from linear random observations has a myriad of applications such as compressed sensing, error correction via sparse superposition codes, Boolean group testing, and supervised machine learning. The fitting of linear relationships among variables in a data set is a standard tool in data analysis. More frequently, there exists conditions under which sparse models fit the data quite well. For example, Rosenfeld et al. \citep{Agrawal} used data to mimic heuristics to identify small segments of a population in which a few additional risk factors were highly predictive of certain kinds of cancer, whereas the same risk factors were not significant in the population \citep{juba2016conditional}. Estimation errors can be characterized by some standard measures in statistics and information theory such as the minimum mean square error (MMSE) and/or mutual information. These fundamental limits are usually obtained by using the replica or interpolation methods in statistical physics where the number of observations is usually assumed to scale linearly with the signal dimension \citep{Edwards1975a}. Accordingly, most of existing approximate message passing algorithms are designed based on the same assumption. However, in many practical applications in machine learning, communications, and signal processing such as medical image recognition and group testing, the number of observations are very small compared with the signal dimension. In this work, we estimate two fundamental limits (MMSE and mutual information) and propose an approximate message passing algorithm to achieve the MMSE for sub-linear regimes where the number of observations scales sub-linearly with the signal dimension. 
\subsection{Related Papers}
In recent years, there has been the progress on a coherent mathematical theory of the replica and interpolation method in statistical physics of spin glasses \citep{Edwards1975a}. These methods have been fruitfully extended and adapted to the problems of interest in a wide range of applications in Bayesian inferences, multiuser communications, and theoretical computer science \citep{Tanaka2002a,GuoShamaiVerdu2005} and \citep{Truong2022OnLM}. The replica method, although very interesting, is based on some non-rigorous assumptions. \citep{Reeves2016TheRP} proved that the replica prediction in \citep{Tanaka2002a,GuoShamaiVerdu2005} is exact. In more recent years, an adaptive interpolation method has been proposed to prove fundamental limits predicted by replica method in a rigorous way \citep{Barbier2017TheAI,pmlr-v75-barbier18a,Barbier2018OptimalEA}. Roughly speaking, this method interpolates between the original problem and the mean-field replica solution in small steps, each step involving its own set of trial parameters and Gaussian mean-fields in the sprit of Guerra and Toninelli \citep{Guerra2002TheTL,Barbier2017TheAI}. We can adjust the set of trial parameters in various ways so that we get both upper and lower bounds that eventually match.

The ``All-or-Nothing" phenomenon for the linear and non-linear models has been characterized in a variety of recent papers. In \citep{pmlr-v99-reeves19a}, Reeves et al. consider a binary $k$-sparse linear regression problem, where the number of observations $m$ is sub-linear to the signal dimension $n$, and established an ``All-or-Nothing" information-theoretic phase transition at a critical sample size $m^*=2k\log(n/k)/\log(1+k/\Delta_n)$ for two regimes $k/\Delta_n=\Omega(1)$ and $k=o(\sqrt{n})$ with $\Delta_n$ being the noise variance. Their results are based on an assumption that the sparse signal is uniformly distributed from the set $\{v \in \{0,1\}^n: \|v\|_0=k\}$. \citep{Reeves2019AllorNothingPF} considers a double limit where one first obtains the high-dimensional limit (under linear sparsity) and then considers the limiting behavior of the RS formulas and the AMP state evolution with respect to a family of prior distributions with which allows the prior to scale with the dimensions. However, the analysis reveals that the resulting formulas can simplify dramatically in the sparse regime. Indeed, in certain cases (e.g., Bernoulli prior) the single-letter mutual information function converges to a piecewise linear limit, and this gives rise to an ``all-or-nothing" phenomenon. 
Sharp information-theoretic bounds were established in \citep{Scarlet2017a} and \citep{TruongS2020} for support recovery problems in linear and phase retrieval models, respectively. In addition, the ``All-or-Nothing" phenomenon was also considered for Bernoulli group testing \citep{Truong2020} or sparse spiked matrix estimation in \citep{Barbier201901PT}, \citep{Luneau2020InformationTL}, and \citep{NilesWeed2020TheAP}. In \citep{Barbier2020}, this phenomenon was also investigated for the generalized linear models with sub-linear regimes and Bernoulli and Bernoulli-Rademacher distributed vectors.

Although the results achieved by the replica method and the adaptive interpolation counterpart are very interesting, they are mainly constrained to the case where the number of observations scales \emph{linearly} with the signal dimension. \citep{Luneau2020InformationTL} considered generalized linear models in regimes where the number of nonzero components of the signal and accessible data points are sublinear with respect to
the size of the signal. They obtained a proof of the replica symmetric formula for the linear model for the case $\alpha>8/9$. In this work, thanks to the development of a new proof technique of the key concentration inequality in \citep{BarbierALT2016}, we can widen the range of $\alpha$ to all $[0,1]$ for a similar model. In addition, we develop a variant of the approximate message passing for the sparse linear regression with sub-linear sparsity which can approach the developed fundamental limit for most of simulation cases. Our numerical results show that the weak recovery (and detection) \citep{pmlr-v99-reeves19a} is possible at various ranges of SNR \emph{under the sparsity in the expected sense} where the \emph{expected} number of nonzero elements, $k$, in the vector $\bS$ is much less than the signal dimension $n$. For the sparse model in \citep{pmlr-v99-reeves19a}, the number of nonzero elements in each vector $\bS$ is always equal to $k$. 

Approximate message passing (AMP) refers to a class of efficient algorithms for statistical estimation in high-dimensional problems such as compressed sensing and low-rank matrix estimation. AMP is initially proposed for sparse signal recovery and compressed sensing \citep{Donoho06,Candes06,Metzler16a}. AMP algorithms have been proved to be effective in reconstructing sparse signals from a small number of incoherent linear measurements. Their dynamics are accurately tracked by a simple one-dimensional iteration termed \emph{state evolution} \citep{BayatiMonta2011}. AMP algorithms achieve state-of-the-art performance for several high-dimensional statistical estimation problems, including compressed sensing \citep{Donoho2009a,BayatiMonta2011,Krzakala2012} and low-rank matrix estimation \citep{Matsushita2013LowrankMR, Deshpande2014, Kabashima2016PhaseTA, Montanari2021EstimationOL}.  Moreover, these techniques are also popular and practical in a variety of engineering and computer science applications such as imaging \citep{Fletcher2014ScalableIF, Metzler2017LearnedDP}, communications \citep{Schniter2011AMR,Rush2017CapacityAchievingSS},  and deep learning \citep{Pandit2019AsymptoticsOM,Emami2020GeneralizationEO,Pandit2020InferenceWD}. See \citep{Feng2022AUT} for a detailed survey on this research topic. 
Our results imply that a judicious modification of AMP for linear regimes can work well for sub-linear ones. 
\subsection{Contributions}
In this paper, we consider the same $k$-sparse linear regression as \citep{pmlr-v99-reeves19a} but in more general signal domain. However, we assume that the signal is sparse in expected sense as \citep{Barbier2020} and the number of observation is sub-linear to the signal dimension. Our contributions include: 
\begin{itemize}
	\item We characterize MMSE and mutual information exactly for the sub-linear regimes where $k=O(n^{\alpha})$ and $m=\delta n^{\alpha}$ for some $\alpha \in (0,1]$. Our result is achieved by a generalization of the adaptive interpolation method in Bayesian inference for linear regimes \citep{BarbierALT2016,Barbier2017TheAI} to sub-linear ones. Compared with \citep{BarbierALT2016, Luneau2020InformationTL}, the bound (RHS) in the concentration in Lemma \ref{beo1} is new. We need to develop a new proof to show this concentration inequality for the sub-linear sparsity. 
	\item We design a variant of the classical AMP algorithm \citep{Donoho2009a} for the sub-linear regimes which approaches the information-theoretic fundamental limits for many cases. The state evolution is also rigorously analysed in our work, and we redefine the states in non-asymptotic sense. As a by-product, we generalize a general version of the strong law of large numbers and H\'{a}jek-R\'{e}nyi type maximal inequality, which may be of independent interest.
	\item We perform some numerical evaluations and show that the gap between MSE achieved by our AMP and the MMSE fundamental limit is very small. Our results also show that the new variant of AMP works well for a wide range of $\alpha$ in $[0,1]$.
\end{itemize}  
\subsection{Paper Organization}
The problem settings is placed in Section \ref{sec:setting}, where we introduce the system model and our assumptions. In Section \ref{sec:limits}, we state some information-theoretic fundamental limits such as the average mutual information and MMSE which are obtained by using a rigorous analysis with the adaptive interpolation method. An approximate message passing is proposed and its performance analysis is given in Section \ref{sec:alg}. We place some auxiliary but important proofs in the appendices. 
\section{Problem Settings} \label{sec:setting}
\subsection{Problem settings}
Let $\bS \in \bbR^n$ be a signal observed via a linear model with measurement matrix $\bA\in \bbR^{m \times n}$. Let $\{\Delta_n\}_{n=1}^{\infty}$ be a positive sequence. We consider the same linear model as  \citep{Barbier2017TheAI}: 
\begin{align}
\bY=\bA \bS +\bW \sqrt{\Delta_n} \label{model:setnew},
\end{align} where $\bA\in \bbR^{m \times n}, \bS=(S_1,S_2,\cdots,S_n)^T \in \bbR^n, \bW \in \bbR^m$, and $\bY \in \bbR^m$. Instead of assuming that $m=n \delta$ for some $\delta>0$ as standard literature in replica and adaptive interpolation methods, we assume that $m= \delta n^{\alpha}$ for some $\alpha>0$ and $0<\alpha\leq 1$. We also assume:
\begin{enumerate}
	\item $\bA$ is a Gaussian matrix with $A_{ij} \sim \calN(0,1/m)$.
	\item $\{S_n\}_{n=1}^{\infty}$ is an i.~i.~d. sequence with $S_i \sim \tilP_0 $, where  $\tilP_0(s)=\big(1-k/n\big)\delta(s)+\big(k/n\big) P_0(s)$ for some $k=O(n^{\alpha})$ with $0<\alpha\leq 1$ and $P_0(s)=\sum_{b=1}^B p_b\delta(s-a_b)$ with a finite number $B$ of constant terms such that $s_{\max}:=\max_{b}|a_b|<\infty$.
	\item $ \bW \sim \calN(0,\bI_m)$.
	\item $\Delta_n$ can be any function of $\alpha$, and $n$.
\end{enumerate}
\subsection{Notations}
For any $k>1$, we say a function $\phi:\bbR^q \to \bbR$ is \emph{pseudo-Lipschitz} of order $k$ if there exists a constant $L>0$ such that for all $x,y \in \bbR^q$:
\begin{align}
\big|\phi(x)-\phi(y)\big|\leq L\big(1+\|x\|^{k-1}+\|y\|^{k-1}\big)\|x-y\|
\end{align} where $\|\|$ denotes the Euclidean norm-$2$. In addition, for any sequence of vectors $\{\bx^{(n)}\}_{n=1}^{\infty}$, we denote by $\bx_1^n=\{\bx^{(1)},\bx^{(2)},\cdots,\bx^{(n)}\}$. As standard literature, the mutual information between two random vectors $\bX$ and $\bY$ is written as $I(\bX;\bY)$. The transpose of a matrix $\bA$ is denoted by $\bA^*$. The $\sigma$-algebra which is generated by the union of two $\sigma$-algebras $\calG_1$ and $\calG_2$ is denoted by $\sigma(\calG_1) \cup \sigma(\calG_2)$.
\section{Information-Theoretic Fundamental Limits} \label{sec:limits}
For this case, we assume that the sensing matrix $\bA$ has i.i.d. Gaussian components. Similar to \citep{Barbier2017TheAI}, let 
\begin{align}
\Sigma(u;v)^{-2}&:=\frac{\delta n^{\alpha-1}}{u+v}, \label{bet1}\\
\psi(u;v)&:=\frac{\delta}{2}\bigg[\log\bigg(1+\frac{u}{v}\bigg)-\frac{u}{u+v}\bigg] \label{bet2}.
\end{align}
Define the following sequence of Replica Symmetric (RS)
potentials:
\begin{align}
f_{n,\rm{RS}}(E;\Delta_n):=\psi(E;\Delta_n)+ i_{n,\rm{den}}(\Sigma(E;\Delta_n)) \label{defnrs},
\end{align} where $i_{n,\rm{den}}(\Sigma)=n^{1-\alpha} I(S;S+\tilW \Sigma)$ is a normalized mutual information of a scalar Gaussian denoising model $Y=S+ \tilW \Sigma$ with $S \sim \tilP_0, \tilW\sim \calN(0,1)$, and $\Sigma^{-2}$ an effective signal to noise ratio:
\begin{align}
i_{n,\rm{den}}(\Sigma):=n^{1-\alpha}\bbE_{S,\tilW}\bigg[\log \int \tilP_0(x) \exp\bigg[-\frac{1}{\Sigma^2}\bigg(\frac{(x-S)^2}{2}-(x-S)\tilW \Sigma\bigg)\bigg]dx \label{defiden}.
\end{align}
Our information-theoretic fundamental result is the following:
\noindent
\begin{theorem}\label{mainthm} Let $\nu_n=n^{\alpha-1}\bbE_{S \sim P_0}[S^2]$ and  $\hatbS=\bbE[\bS|\bY]$ be the MMSE estimator. Then, under the condition that $\Delta_n =\Omega_n(1)$\footnote{This constraint is less strict than the one in \citep{Barbier2017TheAI}} and that $\argmin_{E\in [0,\nu_n]}f_{n,\rm{RS}}(E;\Delta_n )$ is unique for all $\Delta_n$, in the large system limits, the following holds:
	\begin{align}
	&\lim_{n\to \infty} \bigg[\frac{I(\bS;\bY|\bA)}{n^{\alpha}}- \min_{E\in [0,\nu_n]} f_{n,\rm{RS}}(E ;\Delta_n )\bigg]=0, \label{eq1key1}\\
	&\lim_{n\to \infty} \bigg[\frac{1}{n^{\alpha}}\sum_{i=1}^n \bbE[(S_i-\hatS_i)^2]-n^{1-\alpha}\tilE(\Delta_n)\bigg]=0, \label{eq35eqpre1}
	\end{align}
	where $\tilE(\Delta_n)$ is the global minimizer of $\min_{E\in [0,\nu_n]}f_{n,\rm{RS}}(E;\Delta_n )$.
\end{theorem} 
\begin{remark} Some remarks are in order.
	\begin{itemize}
		\item Our theorem allows $\Delta_n$ to be dependent on $n$, where $\Delta_n$ is assumed to be fixed in \citep{Barbier2017TheAI}.
		\item For $\alpha=1$ (or $m=\delta n$), and $\Delta_n=\Delta$ for some fixed $\Delta>0$, our results recover the classical result as in \citep{Tanaka2002a, Guo2005a, Reeves2019,Barbier2018OptimalEA}. In these classical papers, the authors assume that $\{S_n\}_{n=1}^{\infty}$ are i.i.d and $S_1 \sim \tilP_0$ which is a fixed distribution. 
		\item The minimization problem in \eqref{eq1key1}, i.e., $\min_{E\in [0,\nu_n]} f_{n,\rm{RS}}(E ;\Delta_n )$ may have multiple minimizers at some values of $\Delta_n$. The number of minimizers depends on the prior distribution. By limiting $E \in [0,\nu_n]$, we may avoid this phenomenon for some cases. 
		\item In our model, the sparsity is in the expected sense, which is different from the models in \citep[ Eqn.~(3)]{pmlr-v99-reeves19a} or  \citep[Cor.~1]{Scarlet2017a}, where the authors assume that the sparse vector is uniformly distributed over ${n}\choose{k}$ possible sparse vectors. In addition, the equivalent SNR in our model is a fixed constant, but the required SNR for \citep[Theorem 3]{pmlr-v99-reeves19a} to hold is greater than an ambiguous constant. Hence, the weak (strong) recovery is expected to hold at low SNR. In our numerical simulations (cf. Fig. \ref{fig:DMK2}), even the classical AMP can (at least) recover the sparse vector weakly, i.e. the normalized MSE (divided by $n^{\alpha}$) is almost less than one for many ranges of SNR \footnote{We can even verify that the normalized sum of MSE by $n^{\alpha}$ mostly in $[0,1]$ by running the classical AMP in \cite[Section C]{BayatiMonta2011}, an sub-optimal algorithm for this setting.}. 
		\item Our results show that under the MMSE estimator, the linear regression model can be decomposed into sub-AWGN channels, and the normalized MMSE of the model is equal to the MMSE of a (time-varying SNR) sub-AWGN channel in the large system limit.  
	\end{itemize}
\end{remark}
\begin{proof}
	The proof of Theorem \ref{mainthm} is based on \citep{Barbier2017TheAI, BarbierALT2016} with some modifications in concentration inequalities and normalized factors to account for new settings. Given the model \eqref{model:setnew}, the likelihood of the observation $\by$ given $\bS$ and $\bA$ is 
	\begin{align}
	P(\by|\bs,\bA)=\frac{1}{(2\pi \Delta_n)^{m/2}}\exp\bigg[-\frac{1}{2\Delta_n}\big\|\by-\bA \bs\big\|^2\bigg].
	\end{align}
	From Bayes formula we then get the posterior distribution for $\bx=[x_1,x_2,\cdots,x_n] \in \bbR^n$ given the observation $\by$ and sensing matrix $\bA$
	\begin{align}
	P(\bx|\by,\bA)&=\frac{\prod_{i=1}^n \tilP_0(x_i)P(\by|\bx,\bA)}{\int \prod_{i=1}^n	\tilP_0(x_i) dx_i P(\by|\bx,\bA)}.
	\end{align} 
	Replacing the observation $\by$ by its explicit expression \eqref{model:setnew} as a function of the signal and the noise we obtain
	\begin{align}
	&P(\bx|\by=\bA \bs+\bw \sqrt{\Delta_n},\bA)\nn\\
	&=\frac{\prod_{i=1}^n \tilP_0(x_i)e^{-\calH(\bx;\bA,\bs,\bw)}}{\calZ(\bA,\bs,\bw)},
	\end{align}
	where we call 
	\begin{align}
	&\calH(\bx;\bA,\bs,\bw):=\frac{1}{\Delta_n}\sum_{\mu=1}^m\bigg( \frac{1}{2}\bigg[\bA(\bx-\bs)\bigg]_{\mu}^2 -\bigg[\bA(\bx-\bs)\bigg]_{\mu}w_{\mu}\sqrt{\Delta_n}\bigg)
	\end{align} the \emph{Hamiltonian} of the model, and the normalization factor is by definition the \emph{partition function}:
	\begin{align}
	\calZ(\bA,\bs,\bw):=\int\bigg\{\prod_{i=1}^n \tilP_0(x_i)dx_i\bigg\}e^{-\calH(\bx;\bA,\bs,\bw)}.
	\end{align} 
	Our principal quantity of interest is
	\begin{align}
	&f_n=-\frac{1}{n^{\alpha}}\bbE_{\bA,\bS,\bW}[\log \calZ(\bA,\bS,\bW)]\\
	&=-\frac{1}{n^{\alpha}}\bbE_{\bA,\bS,\bW}\bigg[\log\bigg( \int\bigg\{\prod_{i=1}^n \tilP_0(x_i)dx_i\bigg\}\nn\\
	&\qquad \times \exp\bigg(-\frac{1}{\Delta_n}\sum_{\mu=1}^m\bigg( \frac{1}{2}\bigg[\bA(\bx-\bS)\bigg]_{\mu}^2 -\bigg[\bA(\bx-\bS)\bigg]_{\mu}W_{\mu}\sqrt{\Delta_n}\bigg) \bigg)\bigg) \bigg],
	\end{align} where $\bW \stackrel{i.i.d.}{\sim} \calN(0,1)$. 
	
	By using the Bayes' rule
	\begin{align}
	P(\by|\bA)=\frac{P(\by|\bx,\bA)\prod_{i=1}^n \tilP_0(x_i)}{P(\bx|\by=\bA \bs+ \bw \sqrt{\Delta_n},\bA)},
	\end{align} we have
	\begin{align}
	P(\by|\bA)=(2\pi\Delta)^{-m/2}\calZ(\bA,\bs,\bw)e^{-\frac{\|\bw^2\|}{2}}.
	\end{align}
	It follows that
	\begin{align}
	&\frac{I(\bS;\bY|\bA)}{n^{\alpha}}=\frac{1}{n^{\alpha}}\bbE_{\bA,\bS,\bY}\bigg[\log\bigg( \frac{P(\bS,\bY|\bA)}{\tilP_0(\bS)P(\bY|\bA)}\bigg)\bigg]\\
	&=f_n-\frac{h(\bY|\bA,\bS)}{n^{\alpha}}+\frac{1}{2n^{\alpha}}\bbE[\|\bW\|^2]+\frac{m}{2n^{\alpha}}\log(2\pi \Delta_n)\\
	&=f_n-\frac{m}{2n^{\alpha}}\log(2\pi e \Delta_n)+\frac{m}{2n^{\alpha}}+\frac{m}{2n^{\alpha}}\log(2\pi \Delta_n)\\
	&=f_n.
	\end{align} Hence, in order to obtain \eqref{eq1key1}, it is enough to show that 
	\begin{align}
	\lim_{n\to \infty}\big[f_n -\min_{E\in [0,\nu_n] }f_{n,\rm{RS}}(E;\Delta_n)\big]=0 \label{buchi1}.
	\end{align} Let $\bW^{(k)}=[W_{\mu}^{(k)}]_{\mu=1}^m, \tilbW^{(k)}=[\tilW_i^{(k)}]_{i=1}^n$ and $\hatbW=[\hatW_i]_{i=1}^n$ all with i.i.d. $\calN(0,1)$ entries for $k=1,2,\cdots,K_n$ where $K_n$ is chosen later. Define $\Sigma_k:=\Sigma(E_k;\Delta_n)$ where the trial parameters $\{E_k\}_{k=1}^{K_n}$ are determined later on. Given any fixed $\eps \in [0,1]$, as \citep{Barbier2017TheAI}, the (perturbed) $(k,t)$-interpolating Hamiltonian for this problem is defined as
	\begin{align}
	\calH_{k,t;\eps}(\bx;\bTheta)&:=\sum_{k'=k+1}^{K_n} h\bigg(\bx,\bS, \bA,\bW^{(k')},K_n\Delta_n\bigg) +\sum_{k'=1}^{k-1}h_{\rm{mf}}\bigg(\bx,\bS,\tilbW^{(k')},K_n \Sigma_{k'}^2\bigg) \nn\\
	&\qquad  + h\bigg(\bx,\bS,\bA,\bW^{(k)},\frac{K_n}{\gamma_k(t)}\bigg) +h_{\rm{mf}}\bigg(\bx,\bS,\tilbW^{(k)},\frac{K_n}{\lambda_k(t)}\bigg)\nn\\
	&\qquad \qquad +\eps\sum_{i=1}^n\bigg(\frac{x_i^2}{2}-x_iS_i-\frac{x_i \hatW_i}{\sqrt{\eps}}\bigg) \label{defHamil1}.
	\end{align} 
	Here,
	$
	\bTheta:=\{\bS,\bW^{(k)},\tilbW^{(k)}\}_{k=1}^{K_n},\hatbW,\bA, k\in [K_n], t \in [0,1]
	$
	and
	\begin{align}
	h(\bx,\bS,\bW,\bA,\sigma^2)&:=\frac{1}{\sigma^2}\sum_{\mu=1}^m \bigg(\frac{[\bA \bar{\bx}]_{\mu}^2}{2}-\sigma\big[\bA \bar{\bx}\big]_{\mu}W_{\mu}\bigg),\\
	h_{\rm{mf}}(\bx,\bS,\tilbW,\sigma^2)&:=\frac{1}{\sigma^2}\sum_{i=1}^n \bigg(\frac{\barx_i^2}{2}-\sigma \barx_i \tilW_i\bigg),
	\end{align} where $\bar{\bx}=\bx-\bS$ and $\barx_i=x_i-S_i$. 
	
	The $(k,t)$-interpolating model corresponds an inference model where one has access to the following sets of noisy observations about the signal $\bS$ 
	\begin{align}
	\bigg\{\bZ^{(k')}&=\bA \bS + \bW^{(k')}\sqrt{K_n \Delta_n}\bigg\}_{k'=k+1}^{K_n},\\
	\bigg\{\tilbZ^{(k')}&=\bS+\tilbW^{(k')}\Sigma_{k'}\sqrt{K_n}\bigg\}_{k'=1}^{k-1},\\
	\bigg\{\bZ^{(k)}&=\bA \bS +\bW^{(k)}\sqrt{\frac{K_n}{\gamma_k(t)}}\bigg\} \label{obser1},\\
	\bigg\{\tilbZ^{(k)}&=\bS+\tilbW^{(k)}\sqrt{\frac{K_n}{\lambda_k(t)}}\bigg\} \label{obser2}.
	\end{align}
	The first and third sets of observation correspond to similar inference channel as the original model \eqref{model:setnew} but with a higher noise variance proportional to $K_n$. These correspond to the first and third terms in \eqref{defHamil1}. The second and fourth sets instead correspond to decoupled Gaussian denoising models, with associated ``mean-field" second and fourth term in \eqref{defHamil1}. The last term in \eqref{defHamil1} is a perturbed term which corresponds to a Gaussian ``side-channel" $\bY=\bS \sqrt{\eps}+ \hatbZ$ whose signal-to-noise ratio $\eps$ will tend to zero at the end of proof. The noise variance are proportional to $K_n$ in order to keep the average signal-to-noise ratio not dependent on $K_n$. A perturbed of the original and final (decoupled) models are obtained by setting $k=1,t=0$ and $k=K_n,t=1$, respectively. The interpolation is performed on both $k$ and $t$. For each fixed $k$, at $t$ changes from $0$ to $1$, the observation in \eqref{obser1} is removed from the original model and added to the decoupled model. An interesting point is that the $(k,t=1)$ and $(k+1,t=0)$-interpolating models are statistically equivalent. This is an adjusted model of the classical interpolation model in \citep{Guerra2002TheTL}, where an interpolating path $k\in [K_n]$ is added. This is called the adaptive interpolation method. See \citep{Barbier2017TheAI} for more detailed discussion.

	Consider a set of observations $[\by,\tilby]$ from the following channels
	\begin{align}
	\begin{cases}
	\by&=\bA \bS + \bW\frac{1}{\sqrt{\gamma_k(t)}}\\
	\tilby&=\bS+ \tilbW \frac{1}{\sqrt{\lambda_k(t)}},
	\end{cases} 
	\end{align} where $\bW \sim \calN(0,\bI_{m}), \tilbW \sim \calN(0,\bI_n)$, $t \in [0,1]$ is the interpolating parameter and the ``signal-to-noise functions" $\{\gamma_k(t),\lambda_k(t)\}_{k=1}^{K_n}$ satisfy
	\begin{align}
	\gamma_k(0)&=\Delta_n^{-1}, \qquad \gamma_k(1)=0,\\
	\lambda_k(0)&=0, \qquad \lambda_k(1)=\Sigma_k^{-2},
	\end{align} as well as the following constraint
	\begin{align}
	\frac{\delta n^{\alpha-1}}{\gamma_k(t)^{-1}+E_k}+\lambda_k(t)=\frac{\delta n^{\alpha-1}}{\Delta_n +E_k }=\Sigma_k^{-2}
	\end{align}
	and thus
	\begin{align}
	\frac{d\lambda_k(t)}{dt}=-\frac{d\gamma_k(t)}{dt}\frac{\delta n^{\alpha-1}}{(1+\gamma_k(t)E_k)^2} \label{ach1}.
	\end{align}
	We also require $\gamma_k(t)$ to be strictly decreasing with $t$. The $(k,t)$-interpolating model has an associated posterior distribution, Gibbs expectation $\langle - \rangle_{k,t;\eps}$ and $(k,t)$-interpolating free energy $f_{k,t;\eps}$:
	\begin{align}
	&P_{k,t;\eps}(\bx|\bTheta):=\frac{\prod_{i=1}^n \tilP_0(x_i)e^{-\calH_{k,t;\eps}(\bx;\btheta)}}{\int\big\{\prod_{i=1}^n \tilP_0(x_i)\big\}e^{-\calH_{k,t;\eps}(\bx;\btheta)}},\\
	&\langle V(\bX)\rangle_{k,t;\eps}:=\int d\bx V(\bx)P_{k,t;\eps}(\bx|\bTheta),\\
	&f_{k,t;\eps}:=-\frac{1}{n}\bbE_{\bTheta}\bigg[\log \int \bigg\{\prod_{i=1}^n dx_i \tilP_0(x_i)\bigg\}e^{-\calH_{k,t;\eps}(\bx;\bTheta)}\bigg].
	\end{align}
	\begin{lemma} \label{easylem} Let $P_0$ have finite second moment. Then for initial and final systems 
		\begin{align}
		|f_{1,0;\eps}-f_{1,0;0}|&\leq O\bigg(\frac{\eps }{2n^{1-\alpha}}\bigg)\bbE_{S \sim P_0}[S^2]\\ |f_{K_n,1;\eps}-f_{K_n,1;0}|&\leq O\bigg(\frac{\eps }{2n^{1-\alpha}}\bigg)\bbE_{S \sim P_0}[S^2].
		\end{align}
	\end{lemma}
    \begin{proof} Using the similar arguments as Lemma 1, Section II in \citep{Barbier2017TheAI}, we have
		\begin{align}
		|f_{1,0;\eps}-f_{1,0;0}|& \leq \frac{\eps }{2 }\bbE_{S \sim \tilP_0}[S^2]\\
		&=\frac{\eps }{2}\frac{k}{n}\bbE_{S \sim P_0}[S^2]\\
		&=O\bigg(\frac{\eps}{2n^{1-\alpha}}\bigg)\bbE_{S \sim P_0}[S^2].
		\end{align} 
		
		Similarly, we come to the other inequality.
	\end{proof}
	Now, by defining
	\begin{align}
	\Sigma_{\rm{mf}}^{-2}(\{E_k\}_{k=1}^{K_n};\Delta):=\frac{1}{K_n}\sum_{k=1}^{K_n} \Sigma_k^2,
	\end{align}
	from \eqref{defHamil1}, we have
	\begin{align}
	\calH_{K_n,1;0}(\bx;\bTheta)&=\sum_{k=1}^{K_n} h_{\rm{mf}}(\bx,\bS,\bA,\tilbW^{(k)},K_n \Sigma_k^{-2})\\
	&=\sum_{k=1}^{K_n}\frac{1}{K_n \Sigma_k^2} \sum_{i=1}^n \bigg(\frac{\barx_i^2}{2}-\sqrt{K_n \Sigma_k^2} \barx_i \tilW_{\mu}^{(k)}\bigg)\\
	&=\Sigma_{\rm{mf}}^{-2} \bigg(\sum_{i=1}^n \frac{\barx_i^2}{2} -\Sigma_{\rm{mf}} \barx_i \sum_{k=1}^{K_n} \frac{\Sigma_{\rm{mf}}}{\sqrt{K_n \Sigma_k^2}}\tilW_{\mu}^{(k)}\bigg) \label{culi}.
	\end{align}
	Since 
	\begin{align}
	\tilW:=\sum_{k=1}^{K_n} \frac{\Sigma_{\rm{mf}}}{\sqrt{K_n \Sigma_k^2}}\tilW_{\mu}^{(k)} \sim \calN(0,1),
	\end{align} it holds from \eqref{culi} that
	\begin{align}
	&\calH_{K_n,1;0}(\bx;\bTheta)=\Sigma_{\rm{mf}}^{-2} \bigg(\sum_{i=1}^n \frac{\barx_i^2}{2} -\Sigma_{\rm{mf}} \barx_i  \tilW \bigg).
	\end{align}
	Hence, we have
	\begin{align}
	f_{K_n,1;0}&=-\frac{1}{n}\bbE\bigg[\sum_{i=1}^n \log \int dx_i \tilP_0(x_i) e^{-\Sigma_{\rm{mf}}^{-2}\big(\frac{\barx_i^2}{2}-\Sigma_{\rm{mf}} \barx_i \tilW \big)}\bigg]\\
	&=\bbE\bigg[\log \int dx \tilP_0(x) e^{-\Sigma_{\rm{mf}}^{-2}\big(\frac{\barx^2}{2}-\Sigma_{\rm{mf}} \barx \tilW\big)}\bigg]\\
	&=\frac{1}{n^{1-\alpha}}i_{n,\rm{den}}(\Sigma_{\rm{mf}}(\{E_k\}_{k=1}^{K_n};\Delta_n)\label{eq55},
	\end{align} where \eqref{eq55} follows from \eqref{defiden}. 
	
	Similarly, we can show that
	\begin{align}
	f_{1,0;0}&=-\frac{1}{n}\bbE\bigg[\log \int \bigg\{\prod_{i=1}^n dx_i \tilP_0(x_i) e^{-\calH(\bx;\bA,\bS,\bW)}\}\bigg]\\
	&=\frac{f_n}{n^{1-\alpha}} \label{afact}.
	\end{align}
	
	In addition, we can prove (with $\bar{\bX}=\bX-\bS$) that
	\begin{align}
	\frac{df_{k,t;\eps}}{dt}&=\frac{1}{K_n}\big(\calA_{k,t;\eps}+\calB_{k,t;\eps}\big),\\
	\calA_{k,t;\eps}&:=\frac{d\gamma_k(t)}{dt}\frac{1}{2n}\sum_{\mu=1}^m \bbE\bigg[\bigg\langle \big[\bA \bar{\bX} \big]_{\mu}^2 -\sqrt{\frac{K_n}{\gamma_k(t)}}\big[\bA \bar{\bX}\big]_{\mu} W_{\mu}^{(k)}\bigg\rangle_{k,t;\eps}\bigg],\\
	\calB_{k,t;\eps}&:=\frac{d\lambda_k(t)}{dt}\frac{1}{2n}\sum_{i=1}^n \bbE\bigg[\bigg\langle \barX_i^2-\sqrt{\frac{K_n}{\lambda_k(t)}}\barX_i \tilW_i^{(k)}\bigg\rangle_{k,t;\eps}\bigg],
	\end{align} where $\bE$ denotes the average w.r.t. $\bX$ and all quenched random variables $\bTheta$, and $\langle - \rangle_{k,t;\eps}$ is the Gibbs average with Hamiltonian \eqref{defHamil1}. 
	
	Now, since $W_{\mu}^{k}\sim \calN(0,1)$, by using the Gaussian integral formula $\bbE[Zf(Z)]=\bbE[f'(Z)]$, we can show that
	\begin{align}
	n^{1-\alpha}\calA_{k,t;\eps}&=\frac{d\gamma_k(t)}{dt}\frac{1}{2n^{\alpha}}\sum_{\mu=1}^m \bbE\big[\big\langle\big[\bA \bar{\bX}\big]_{\mu}\big\rangle_{k,t;\eps}^2\big]\\
	&=\frac{d \gamma_k(t)}{dt}\frac{\delta}{2}\rm{ymmse}_{k,t;\eps} \label{eq66},
	\end{align} where
	\begin{align}
	{\rm{ymmse}_{k,t;\eps}}:=\frac{1}{m}\bbE\big[\big\|\bA (\langle \bX \rangle_{k,t;\eps}-\bS)\big\|^2\big] \label{Q1}
	\end{align} is called ``measurement minimum mean-square error".
	
	For $\calB_{k,t;\eps}$, we proceed similarly and find
	\begin{align}
	n^{1-\alpha}\calB_{k,t;\eps}&=\frac{d\lambda_k(t)}{dt}\frac{1}{2n^{\alpha}}\sum_{i=1}^n \bbE[\langle \barX_i^2\rangle_{k,t;\eps}]\\
	&=\frac{d\lambda_k(t)}{dt}\frac{1}{2n^{\alpha}}\bbE[\|\langle \bX\rangle_{k,t;\eps}-\bS\|^2]\\
	&=-\frac{d\gamma_k(t)}{dt}\frac{1}{(1+\gamma_k(t)E_k)^2}\frac{\delta}{2} n^{\alpha-1}\rm{mmse}_{k,t;\eps} \label{eq67},
	\end{align} where the normalized minimum mean-square-error (MMSE) defined as
	\begin{align}
	{\rm{mmse}_{k,t;\eps}}:=\frac{1}{n^{\alpha}}\bbE[\|\langle \bX\rangle_{k,t;\eps}-\bS\|^2].
	\end{align} Here, \eqref{eq67} follows from \eqref{ach1}.

	By the construction, we have the following coherency property: The $(k,t=1)$ and $(k+1,t=0)$ models are equivalent (the Hamiltonian is invariant under this change) and thus $f_{k,1;\eps}=f_{k+1,0;\eps}$ for any $k$ \citep{Barbier2017TheAI}. This implies that the $(k,t)$-interpolating free energy satisfies
	\begin{align}
	f_{1,0;\eps}&=f_{K_n,1;\eps}+\sum_{k=1}^{K_n} (f_{k,0;\eps}-f_{k,1;\eps})\\
	&=f_{K_n,1;\eps}-\sum_{k=1}^{K_n} \int_0^1 dt \frac{df_{k,t;\eps}}{dt}.
	\end{align}
	It follows that
	\begin{align}
	f_{1,0;\eps}n^{1-\alpha} &= n^{1-\alpha} f_{K_n,1;\eps}-n^{1-\alpha}\sum_{k=1}^{K_n} \int_0^1 dt \frac{df_{k,t;\eps}}{dt}\\
	&=n^{1-\alpha} f_{K_n,1;\eps}-\frac{1}{K_n}\sum_{k=1}^{K_n}\int_0^1 dt\bigg(n^{1-\alpha}\calA_{k,t;\eps}+n^{1-\alpha}\calB_{k,t;\eps}\bigg) \label{eq68}.
	\end{align}
	On the other hand, by Lemma \ref{easylem}, we have
	\begin{align}
	n^{1-\alpha}|f_{1,0,0}-f_{1,0;\eps}|\leq \frac{\eps}{2}\bbE_{S \sim P_0} \bbE[S^2] \label{eq68b},
	\end{align}
	or
	\begin{align}
	|f_n-n^{1-\alpha}f_{1,0;\eps}|\leq \frac{\eps}{2}\bbE_{S \sim P_0} \bbE[S^2] \label{eq68c}
	\end{align} where \eqref{eq68c} follows from \eqref{afact}. 
		
	From \eqref{eq66}, \eqref{eq67}, and \eqref{eq68}, we obtain
	\begin{align}
	&\int_{a_n}^{b_n} d\eps f_n =n^{1-\alpha}\int_{a_n}^{b_n}d\eps f_{1,0;\eps}\pm \frac{\eps}{2}\bbE_{S \sim P_0} \bbE[S^2] \label{eq73}\\
	&\qquad =\int_{a_n}^{b_n}n^{1-\alpha}d\eps\bigg\{f_{K_n,1;\eps}-f_{K_n,1;0}\bigg\}+ \int_{a_n}^{b_n}d\eps  i_{n,\rm{den}}\bigg(\Sigma_{\rm{mf}}\big(\{E_k\}_{k=1}^{K_n};\Delta\big)\bigg)\nn\\
	&\qquad-\frac{\delta }{2}\int_{a_n}^{b_n}d\eps\frac{1}{K_n}\sum_{k=1}^{K_n}\int_0^1dt\frac{d \gamma_k(t)}{dt} \bigg({\rm{ymmse}}_{k,t;\eps}-\frac{{\rm{mmse}}_{k,t;\eps} n^{\alpha-1}}{(1+\gamma_k(t)E_k)^2}\bigg)\pm \frac{\eps}{2}\bbE_{S \sim P_0} \bbE[S^2] \label{cuchi1},
	\end{align} where \eqref{eq73} follows from \eqref{eq68c}.
	
	The following lemma can be verified to hold for the new settings:
	\begin{lemma}\label{beo1} For any sequence $K_n \to +\infty$ and $0<a_n<b_n<1$ (that tend to zero slowly enough in the application), and trial parameters $\{E_k=E_k^{(n)}(\eps)\}_{k=1}^{K_n}$ which are differentiable, bounded and non-increasing in $\eps$, we have
		\begin{align}
		&\int_{a_n}^{b_n}\frac{1}{K_n}\sum_{k=1}^{K_n}\int_0^1 dt \frac{d\gamma_k(t)}{dt}\bigg({\rm{ymmse}_{k,t;\eps}} -\frac{\rm{mmse}_{k,t;\eps}}{n^{1-\alpha}+\gamma_k(t)\rm{mmse}_{k,t;\eps}}  \bigg)\nn\\
		&\qquad=O\bigg(\max\bigg\{o\bigg(\frac{b_n-a_n}{\Delta_n}\bigg), a_n^{-2}n^{-\gamma}\bigg\}\bigg) \label{keyconcen}
		\end{align} as $n\to \infty$ for some $0<\gamma<1$.
	\end{lemma}
	\begin{proof}
		This lemma is a generalization of (93) in \citep{Barbier2017TheAI}. The proof of this lemma can be found in Appendix A. 
	\end{proof}
In addition, the following fact can be proved (see Appendix B). 
	\begin{lemma} \label{basiscexnew} The following holds:
		\begin{align}
		&\big|{\rm{mmse}}_{k,t;\eps}-{\rm{mmse}}_{k,0;\eps}\big|
		=O\bigg(\frac{n^{\frac{3}{2}(1+\alpha)}}{K_n \Delta_n}\bigg).
		\end{align}
	\end{lemma} 
Based on the proof of Lemma \ref{beo1}, another interesting fact can also be derived.
\begin{lemma}\label{lemaat}
\begin{align}
\rm{ymmse}_{1,0;0}=\frac{{\rm{mmse}_{1,0;0}}n^{\alpha-1}}{1+{\rm{mmse}_{1,0;0}} n^{\alpha-1}/\Delta_n}+o_n(1) \label{cure1a}.
\end{align} 
\end{lemma}
\begin{proof}
We can obtain \eqref{cure1a} by setting $(k=1,t=0,\eps=0)$ in Eq.~\eqref{cubest} of the proof of Lemma \ref{beo1} with noting that $\gamma_k(0)=\Delta_n^{-1}$. 
\end{proof}
	Return to the proof of our main theorem. It is easy to verify the following identity:
	\begin{align}
	\psi(E_k;\Delta_n):=\frac{\delta}{2}\int_0^1 dt \frac{d \gamma_k(t)}{dt}\bigg(\frac{E_k}{(1+\gamma_k(t)E_k)^2}-\frac{E_k}{1+\gamma_k(t)E_k}\bigg).
	\end{align} Let
	\begin{align}
	\tilf_{n,\rm{RS}}(\{E_k\}_{k=1}^{K_n};\Delta_n):=i_{n,\rm{den}}\bigg(\Sigma_{\rm{mf}}\{E_k\}_{k=1}^{K_n};\Delta_n\bigg) +\frac{1}{K_n}\sum_{k=1}^{K_n}\psi(E_k;\Delta_n) \label{defpsi}.
	\end{align}
	From \eqref{cuchi1} and Lemma \ref{beo1}, we obtain as $a_n \to 0$ that
	\begin{align}
	\int_{a_n}^{2a_n} d\eps f_n 
	&=\int_{a_n}^{2a_n}n^{1-\alpha}d\eps\bigg\{f_{K_n,1;\eps}-f_{K_n,1;0}\bigg\} + \int_{a_n}^{2a_n} d\eps  \tilf_{n,\rm{RS}}(\{E_k\}_{k=1}^{K_n};\Delta_n)\nn\\
	&\quad + \frac{\delta }{2K_n}\int_{a_n}^{2a_n} d\eps \sum_{k=1}^{K_n}\int_0^1 dt \frac{d\gamma_k(t)}{dt} \frac{\gamma_k(t)(E_k-{\rm{mmse}_{k,t;\eps}} n^{\alpha-1})^2}{(1+\gamma_k(t)E_k)^2(1+\gamma_k(t){\rm{mmse}_{k,t;\eps}}n^{\alpha-1})}\nn\\
	&\qquad+ O\bigg(\max\bigg\{o\bigg(\frac{a_n}{\Delta_n}\bigg), a_n^{-2}n^{-\gamma}\bigg\}\bigg) \label{fundeq11}.
	\end{align}
	Now, since $\gamma_k(t)$ is non-creasing in $t \in [0,1]$, it holds that $\frac{d\gamma_k(t)}{dt}\leq 0$. Hence, from \eqref{fundeq11}, we obtain
	\begin{align}
	\int_{a_n}^{2a_n} d\eps f_n &\leq \int_{a_n}^{2a_n}n^{1-\alpha}d\eps\bigg\{f_{K_n,1;\eps}-f_{K_n,1;0}\bigg\}\nn\\
	&\quad + \int_{a_n}^{2a_n} d\eps  \tilf_{n,\rm{RS}}(\{E_k\}_{k=1}^{K_n};\Delta_n)+ O\bigg(\max\bigg\{o\bigg(\frac{a_n}{\Delta_n}\bigg), a_n^{-2}n^{-\gamma}\bigg\}\bigg) \label{upper}.
	\end{align}
	By setting $E_k=\argmin_{E\in [0,\nu_n]} f_{n,\rm{RS}}(E;\Delta_n)$ for all $k \in [K_n]$, from \eqref{upper}, we have
	\begin{align}
	f_n \leq O(1) \eps+ \min_{E\in [0,\nu_n]} f_{n,\rm{RS}}(E;\Delta_n)+O\bigg(\max\bigg\{o\bigg(\frac{1}{\Delta_n}\bigg), a_n^{-3}n^{-\gamma}\bigg\}\bigg)
	\end{align} for some $\gamma>0$. By taking $\eps \to 0$, we can achieve the following upper bound:
	\begin{align}
	f_n -\min_{E\in [0,\nu_n]} f_{n,\rm{RS}}(E;\Delta_n) \leq O\bigg(\max\bigg\{o\bigg(\frac{1}{\Delta_n}\bigg), a_n^{-3}n^{-\gamma}\bigg\}\bigg) \label{upperbd}.
	\end{align}
	On the other hand, from Lemma \ref{basiscexnew} and \eqref{fundeq11}, by choosing $K_n=\Omega(n^b)$ for some sufficient large $b$ such that $|{\rm{mmse}}_{k,t;\eps}-{\rm{mmse}}_{k,0;\eps}|$ decays sufficiently fast, we obtain
	\begin{align}
	\int_{a_n}^{2a_n} d\eps f_n 
	&=\int_{a_n}^{2a_n}n^{1-\alpha}d\eps\bigg\{f_{K_n,1;\eps}-f_{K_n,1;0}\bigg\} + \int_{a_n}^{2a_n} d\eps  \tilf_{n,\rm{RS}}(\{E_k\}_{k=1}^{K_n};\Delta_n)\nn\\
	&\quad + \frac{\delta}{2K_n}\int_{a_n}^{2a_n} d\eps \sum_{k=1}^{K_n}\int_0^1 dt \frac{d\gamma_k(t)}{dt} \frac{\gamma_k(t)(E_k-{\rm{mmse}_{k,0;\eps}} n^{\alpha-1})^2}{(1+\gamma_k(t)E_k)^2(1+\gamma_k(t){\rm{mmse}_{k,0;\eps}}n^{1-\alpha})}\nn\\
	&\qquad+O\bigg(\max\bigg\{o\bigg(\frac{a_n}{\Delta_n}\bigg), a_n^{-2}n^{-\gamma}\bigg\}\bigg) \label{fundeq11a}.
	\end{align}
	By choosing $E_k={\rm{mmse}_{k,0;\eps}} n^{\alpha-1}$ for all $k \in [K_n]$, it holds that
	\begin{align}
	E_k&={\rm{mmse}_{1,0;\eps}} n^{\alpha-1} \\
	&=\frac{1}{n^{\alpha}}\bbE\big[\|\bS-\bbE[\bS|\bY]\|^2\big] n^{\alpha-1}\\
	&\leq \frac{1}{n^{\alpha}} \bbE[\|\bS\|^2] n^{\alpha-1} \label{buchi}\\
	&= \frac{n}{n^{\alpha}}  \bbE_{S \sim \tilP_0} [S^2] n^{\alpha-1}\\
	&= \frac{n}{n^{\alpha}}  \frac{n^{\alpha}}{n} \bbE_{S \sim P_0} [S^2] n^{\alpha-1}\\
	&=\bbE_{S \sim P_0} [S^2] n^{\alpha-1}\\
	&=\nu_n \label{bestchi},
	\end{align} where $\nu_n$ is defined in Theorem \ref{mainthm}. Here, \eqref{buchi} follows from the fact that MMSE estimation gives the lowest MSE.
	
	Hence, from \eqref{fundeq11a} we have
	\begin{align}
	\int_{a_n}^{2a_n} d\eps f_n& 
	=\int_{a_n}^{2a_n}n^{1-\alpha}d\eps\bigg\{f_{K_n,1;\eps}-f_{K_n,1;0}\bigg\}\nn\\
	&\quad + \int_{a_n}^{2a_n} d\eps  \tilf_{n,\rm{RS}}(\{E_k\}_{k=1}^{K_n};\Delta_n)+O\bigg(\max\bigg\{o\bigg(\frac{a_n}{\Delta_n}\bigg), a_n^{-2}n^{-\gamma}\bigg\}\bigg) \label{fundeq11b}
	\end{align} 
	Now, let $\Sigma_k^{-2}:=\delta n^{\alpha-1}/(E_k+\Delta_n)$ for all $k \in [K_n]$, then it holds that $\Sigma_k^{-2}\geq \frac{\delta n^{\alpha-1}}{\nu_n+\Delta_n}$ by \eqref{bestchi} . For a given $\Delta_n$, set $\psi_{\Delta_n}(\Sigma^{-2})=\psi(\delta n^{\alpha-1}/\Sigma^{-2}-\Delta_n;\Delta_n)$. Since $\psi_{\Delta_n}(\cdot)$ is a convex function, from \eqref{defpsi}, we are easy to see that
	\begin{align}
	\tilf_{n,\rm{RS}}(\{E_k\}_{k=1}^{K_n};\Delta_n)&=i_{n,\rm{den}}\bigg(\Sigma_{\rm{mf}}\bigg(\{E_k\}_{k=1}^{K_n};\Delta_n\bigg)\bigg) +\frac{1}{K_n}\sum_{k=1}^{K_n}\psi_{\Delta}(\Sigma_k^{-2})\\
	&\quad \geq i_{n,\rm{den}}\bigg(\Sigma_{\rm{mf}}\bigg(\{E_k\}_{k=1}^{K_n};\Delta_n\bigg)\bigg)+ \psi_{\Delta_n}\bigg( \Sigma_{\rm{mf}}^{-2}\{E_k\}_{k=1}^{K_n};\Delta_n\bigg) \\
	&\quad \geq \min_{\Sigma\in \big[0,\sqrt{\frac{\nu_n+\Delta_n}{\delta n^{\alpha-1}}}\big]} \bigg(i_{n,\rm{den}}(\Sigma)+ \psi_{\Delta_n}(\Sigma^{-2})\bigg)\\
	&\quad \geq \min_{E\in [0,\nu_n]} f_{n,\rm{RS}}(E;\Delta_n) \label{eq97}
	\end{align} 
	From \eqref{fundeq11b} and \eqref{eq97}, we obtain a lower bound
	\begin{align}
	f_n \geq \min_{E\in [0,\nu_n]} f_{n,\rm{RS}}(E;\Delta_n) +O(1)\eps+O\bigg(\max\bigg\{o\bigg(\frac{1}{\Delta_n}\bigg), a_n^{-3}n^{-\gamma}\bigg\}\bigg) \label{lowerbd}.
	\end{align}
	From \eqref{upperbd} and \eqref{lowerbd}, we have
	\begin{align}
	f_n - \min_{E\in [0,\nu_n]} f_{n,\rm{RS}}(E;\Delta_n)=O\bigg(\max\bigg\{o\bigg(\frac{1}{\Delta_n}\bigg), a_n^{-3}n^{-\gamma}\bigg\}\bigg), 
	\end{align}
	or
	\begin{align}
	\frac{I(\bS;\bY|\bA)}{n^{\alpha}}- \min_{E\in [0,\nu_n]} f_{n,\rm{RS}}(E ;\Delta_n) = O\bigg(\max\bigg\{o\bigg(\frac{1}{\Delta_n}\bigg), a_n^{-3}n^{-\gamma}\bigg\}\bigg) \label{green1} 
	\end{align} which leads to \eqref{buchi1} by choosing the sequence $a_n \to 0$ and $a_n^{-3} n^{-\gamma} \to 0$. 
	
	On the other hand, by Lemma \ref{lemaat}, we have the following relation:
	\begin{align}
	\rm{ymmse}_{1,0;0}=\frac{{\rm{mmse}_{1,0;0}}n^{\alpha-1}}{1+{\rm{mmse}_{1,0;0}} n^{\alpha-1}/\Delta_n}+o_n(1). \label{cure1}
	\end{align} 
	
	Now, denote by
	\begin{align}
	\tili_{n,\rm{den}}(\Sigma)&:= I(S;S+\tilW \Sigma),\\
	\tSigma(E(\Delta_n);\Delta_n)^{-2}&=\frac{\delta}{E(\Delta_n)+\Delta_n},
	\end{align} then from \eqref{bet1} and \eqref{bet2}, we have
	\begin{align}
	i_{n,\rm{den}}(\Sigma)&=n^{1-\alpha}\tili_{n,\rm{den}}(\Sigma) \label{bet3},\\
	\Sigma(E(\Delta_n);\Delta_n)^{-2}&:=n^{\alpha-1}\tSigma(E(\Delta_n);\Delta_n)^{-2} \label{bet4}.
	\end{align}
	Then, from \eqref{defnrs}, we have
	\begin{align}
	\frac{df_{n,\rm{RS}}(\tilE(\Delta_n);\Delta_n)}{d\Delta_n^{-1}}&=\frac{d\psi(\tilE(\Delta_n);\Delta_n)}{d\Delta_n^{-1}}+\frac{di_{n,\rm{den}}(\Sigma(\tilE(\Delta_n);\Delta_n)}{d\Delta_n^{-1}} \\
	&=\frac{d\psi(\tilE(\Delta_n);\Delta_n)}{d\Delta_n^{-1}}+ \bigg(\frac{di_{n,\rm{den}}(\Sigma(\tilE(\Delta_n);\Delta_n)}{\Sigma(\tilE(\Delta_n);\Delta_n)^{-2}}\bigg)\bigg(\frac{d\Sigma(\tilE(\Delta_n);\Delta_n)^{-2}}{d\Delta_n^{-1}}\bigg) \\
	&=\frac{d\psi(\tilE(\Delta_n);\Delta_n)}{d\Delta_n^{-1}}+ \bigg(\frac{d\tili_{n,\rm{den}}(\Sigma(\tilE(\Delta_n);\Delta_n)}{\Sigma(\tilE(\Delta_n);\Delta_n)^{-2}}\bigg)\bigg(\frac{d\tSigma(\tilE(\Delta_n);\Delta_n)^{-2}}{d\Delta_n^{-1}}\bigg) \\
	&=\frac{\delta}{2}\bigg(\frac{\tilE(\Delta_n)}{1+\tilE(\Delta_n)/\Delta_n}\bigg) \label{last},
	\end{align} where \eqref{last} follows from \citep{BarbierALT2016}. 
	
    Hence, for $n \to \infty$, we have 
	\begin{align}
	{\rm{ymmse}_{1,0;0}}&=\frac{1}{\delta n^{\alpha}}\frac{dI(\bS;\bY|\bA)}{d\Delta_n^{-1}} \label{culi1}\\
	&=\frac{2}{\delta}\bigg(\frac{df_{n,\rm{RS}}(\tilE(\Delta_n);\Delta_n)}{d\Delta_n^{-1}}\bigg) \label{eq99} \\
	&=\frac{\tilE(\Delta_n)}{1+\tilE(\Delta_n)/\Delta_n} \label{cure2},
	\end{align} where \eqref{culi1} follows from \citep{GuoShamaiVerdu2005} and \eqref{Q1} with $m=\delta n^{\alpha}$, \eqref{eq99} follows from \eqref{green1}, and \eqref{cure2} follows from \eqref{last}.

	From \eqref{cure1} and \eqref{cure2}, we obtain \eqref{eq35eqpre1}. This concludes our proof of Theorem \ref{mainthm}.
\end{proof}
\section{Algorithm and performance guarantee} \label{sec:alg}
In this section, we propose a way to modify the Approximate Message Passing (AMP) algorithm \citep{BayatiMonta2011} to make it work for sub-linear regimes. See our following Algorithm 1. Compared with the AMP in \citep{BayatiMonta2011}, we multiply an extra term $n^{1-\alpha}$ in the steps 1 (initialize) and 5 (state evolution).  

\begin{algorithm}[tbh]\label{ALG1}
	\caption{AMP for sub-linear regimes.}
	\label{alg:example}
	\begin{algorithmic}
		\STATE {\bfseries Input:} observation $\by$, matrix sizes $m, n$, other parameters $\alpha, \delta$, number of iterations ${\rm{itermax}}$, $t=1$, $U_0 \sim \tilP_0, W \sim \calN(0,1)$, $\{\eta_t\}_{t=1}^{\rm{itermax}}$ are given Lipschitz continuous functions.
		\REPEAT
		\STATE Initialize $\tau = \sqrt{\Delta_n +n^{1-\alpha}\bbE_{S\sim \tilP_0}[S^2]/\delta}, \bz=\b0, \hatbx=\b0, d=0$.
		\STATE $\bz \gets \by-\bA \hatbx + \frac{1}{\delta}\bz d$
		\STATE $\bh \gets \bA^* \bz+ \hatbx$
		\STATE $\hatbx \gets \eta_t(\bh,\tau)$,\quad $d\gets \rm{Mean} \big(\frac{d\eta_t}{dx} (\bh,\tau)\big)$
		\STATE $\tau \gets \sqrt{\Delta_n + (n^{1-\alpha}/\delta)\bbE[(\eta_t(U_0+ \tau W,\tau)-U_0
			)^2]}   $
		\STATE $t \gets t+1$
		\UNTIL{t=\rm{itermax}}
		\STATE {\bfseries Output:} $\hatbx$.
	\end{algorithmic}
\end{algorithm}
From now on, we denote by $\hatbx^{(t)}, \by^{(t)}, \bz^{(t)}, \bh^{(t)},\tau_t$ the value of $\hatbx, \by, \bz, \bh,\tau$ at the iteration $t$, respectively. 

First, we prove a few important lemmas. We begin with a generalization of the general law of large numbers in \citep[Theorem 2.1]{Fazekas2001AGA}. Our generalization may be of independent interest.
\begin{lemma}  \label{GSLLN} Let $\{b_n\}_{n=1}^{\infty}$ be a nondecreasing unbounded sequence of positive numbers, and $\{S_n\}_{n=1}^{\infty}$ be a sequence of random variables. Let $\{\nu_n\}_{n=1}^{\infty}$ be nonnegative numbers. Let $r>0$ and $\rho\geq 0$ be two fixed numbers. Assume that for each $n\geq 1$ 
	\begin{align}
	\bbE\bigg[\max_{1\leq l\leq n} \big|S_l\big|^r\bigg]\leq d_n^{1/(1+\rho)} \sum_{i=1}^n \nu_i \label{cond1}
	\end{align} for some positive and non-increasing sequence  $\{d_n\}_{n=1}^{\infty}$ with $d_1=1$. Under the condition that $d_n b_n \to \infty$ and
	\begin{align}
	\sum_{l=1}^{\infty} \frac{\nu_l}{b_l^r d_l^{r-1/(1+\rho)}} <\infty \label{assum1}, 
	\end{align}
	then
	\begin{align}
	\lim_{n\to \infty} \frac{S_n}{d_n b_n}=0, \qquad \mbox{a.s.} \label{G2}.
	\end{align}
\end{lemma}
\begin{remark} For $d_n=1$ for all $n$, this lemma recovers the general strong law of large numbers in \citep[Theorem 2.1]{Fazekas2001AGA}.
\end{remark}
\begin{proof}
	See Appendix C for a detailed proof.	
\end{proof}
Next, we recall the following lemma from \citep{BayatiMonta2011}.
\begin{lemma} \label{lem:lipschitz} Let $\phi: \bbR^q \to \bbR$ be a pseudo-Lipschitz of order $k$, then the following hold:
	\begin{itemize}
		\item There is a constant $L'$ such that for all $\bx \in \bbR^q: |\phi(\bx)|\leq L'(1+\|\bx\|^k)$.
		\item $\phi$ is locally Lipschitz, that is for any $Q>0$,  there exists a constant $L_{Q,q}<\infty$ such that for all $\bx, \by \in [-Q,Q]^q$,
		\begin{align}
		|\phi(\bx)-\phi(\by)|\leq L_{Q,q}\|\bx-\by\|.
		\end{align}
Further, $L_{Q,q}\leq c[1+(Q\sqrt{q})^{k-1}]$ for some constant $c$.
	\end{itemize}
\end{lemma}
Now, we prove the following important lemma.
\begin{lemma} \label{thm:chang} Let $\{f_t\}_{t\geq 0}$ and $\{g_t\}_{t\geq 0}$ be two sequences of functions, where for each $t \in \bbZ^+$, $f_t:\bbR^2 \to \bbR$ and $g_t:\bbR^2 \to \bbR$ are assumed to be Lipschitz continuous. Given $\tilbw \in \bbR^m$ and $\bs_0 \in \bbR^n$, define the sequence of vectors $\bh^{(t)}, \bq^{(t)} \in \bbR^n$ and $\bb^{(t)},\boldm^{(t)} \in \bbR^m$ such that
	\begin{align}
	\bh^{(t+1)}&=\bA^* \boldm^{(t)}-\zeta_t \bq^{(t)},\qquad  \boldm^{(t)}=g_t(\bb^{(t)},\tilbw) \label{defbm}\\
	\bb^{(t)}&=\bA \bq^{(t)}-\lambda_t \boldm^{(t-1)}, \quad \bq^{(t)}=f_t(\bh^{(t)},\bs_0) \label{deffbq},
	\end{align} where
	$
	\zeta_t=\langle g_t'(\bb^t,\tilbw)\rangle, \lambda_t=\frac{1}{\delta} \langle f_t'(\bh^{(t)},\bs_0)\rangle
	$ (both derivatives are with respect to the first argument.)
	Assume that 
	\begin{align}
	\sigma_0^2=\frac{n^{1-\alpha}}{ \delta}\bigg(\frac{\bbE\big[\|\bq^{(0)}\|^2\big]}{n}\bigg)
	\end{align} is positive and finite, for a sequence of initial conditions of increasing dimensions. State evolution defines quantities $\{\tau_t^2\}_{t\geq 0}$ and $\{\sigma_t^2\}_{t\geq 0}$ via
	\begin{align}
	\tau_t^2&=\bbE\big[g_t(\sigma_t Z,\tilW)^2\big] \label{rectaut},\\
	\sigma_t^2&=\frac{n^{1-\alpha}}{\delta}\bbE\big[f_t\big(\tau_{t-1}Z,U_0\big)^2\big] \label{recursigmat}
	\end{align} where $\tilW \sim \calN(0,\Delta_n)$ and $U_0 \sim \tilP_0$ which are independent of $Z \sim \calN(0,1)$. Then, for any pseudo-Lipschitz function $\psi: \bbR^2 \to \bbR_+$ of order $k$ and $t\geq 0$, it holds almost surely that
	\begin{align}
	\lim_{n \to \infty} \frac{1}{n^{\alpha}} \sum_{i=1}^n \phi_h(h_i^{(1)},h_i^{(2)},\cdots, h_i^{(t+1)},s_{0,i})-n^{1-\alpha} \bbE\bigg[\phi_h(\tau_0 Z_0,\tau_1 Z_1,\cdots, \tau_t Z_t, U_0)\bigg]=0 \label{bat1},\\
	\lim_{m \to \infty} \frac{1}{m} \sum_{i=1}^m \phi_b(b_i^{(1)},b_i^{(2)},\cdots, b_i^{(t+1)},\tilw_i)-\bbE\bigg[\phi_b(\sigma_0 \hatZ_0,\sigma_1 \hatZ_1,\cdots, \sigma_t \hatZ_t,\tilW)\bigg]=0 \label{bat2},
	\end{align} where $(Z_0,Z_1,\cdots,Z_t)$ and $(\hatZ_0,\hatZ_1,\cdots,\hatZ_t)$ are two zero-mean Gaussian vectors independent of $U_0$ and $W$, with $Z_i, \hatZ_i \sim \calN(0,1)$. 
\end{lemma}
\begin{remark} Some remarks are in order.
	\begin{itemize}	
		\item The proof is based on \citep[Theorem 1]{BayatiMonta2011}. As in the converse proof, we need to make use of the sparsity of the signal to achieve \eqref{bat1} and \eqref{bat2}. Two equations \eqref{bat1} and \eqref{rectaut} are the main differences between the proof of Lemma \ref{thm:chang} and the proof of  \citep[Theorem 1]{BayatiMonta2011}. To show this fact, we need to use Lemma \ref{GSLLN} instead of \citep[Theorem 3]{BayatiMonta2011}.
		\item The states $\{\tau_t\}$ are defined in non-asymptotic sense. This means that we allow them to depend on $n$. In \citep{BayatiMonta2011}, all states are defined in the asymptotic sense.
		\item Compared with \citep[Theorem 3]{BayatiMonta2011}, we constraint the set of pseudo-Lipschitz functions with co-domain $\bbR_+$ instead of $\bbR$. This is likely caused by our proof technique.
	\end{itemize}
\end{remark}
\begin{proof}
	The proof of Lemma \ref{thm:chang} is based on \citep[Proof of Theorem 3]{BayatiMonta2011} with some important changes to account for the new settings. In the following, we outline the proof and present all these changes.
	\begin{itemize}
		\item \emph{Step 1:} Let $\calF_{t_1,t_2}^{(n)}:=\sigma(\bb_0^{t_1-1},\boldm_0^{t_1-1},\bh_1^{t_2}, \bq_0^{t_2},\tilbw)$, which is the $\sigma$-algebra generated by all random variables in the bracket. Note that this $\sigma$-algebra is slightly different from the $\sigma$-algebra in \citep[Proof of Theorem 2]{BayatiMonta2011} since it does not covers $\bs_0$. 
		
		Let $\boldm_{\|}^{(t)}$ and $\bq_{\|}^{(t)}$ be orthogonal projections of $\boldm^{(t)}$ and $\bq^{(t)}$ onto $\rm{span}(\boldm^{(0)},\cdots, \boldm^{(t-1)})$ and $\rm{span}(\bq^{(0)},\cdots, \bq^{(t-1)})$, respectively. Then, we can express
		\begin{align}
		\boldm_{\|}^{(t)}=\sum_{i=1}^{t-1}\alpha_i \boldm^{(i)}, \qquad \bq_{\|}^{(t)}=\sum_{i=1}^{t-1}\beta_i \bq^{(i)}
		\end{align} for some tuples $(\alpha_1,\alpha_2,\cdots,\alpha_{t-1})$ and $(\beta_1,\beta_2,\cdots,\beta_{t-1})$, respectively. Define $\boldm_{\perp}^{(t)}=\boldm^{(t)}-\boldm_{\|}^{(t)}$ and $\bq_{\perp}^{(t)}=\bq^{(t)}-\bq_{\|}^{(t)}$.  Then, it can be shown that \citep{BayatiMonta2011}:
		\begin{align}
		\bh^{(t+1)}\big|\calF_{t+1,t}^{(n)}\cup \sigma(\bs_0) &\stackrel{(d)}{=}\sum_{i=0}^{t-1}\alpha_i \bh^{(i+1)}+\tilbA^* \boldm_{\perp}^{(t)}+ \tilbQ_{t+1}\vec{o}_{t+1}(1) \label{fact1}\\
		\bb^{(t)}\big|\calF_{t,t}^{(n)}\cup \sigma(\bs_0)&\stackrel{(d)}{=}\sum_{i=0}^{t-1}\beta_i\bb^{(i)}+\tilbA \bq_{\perp}^{(t)}+ \tilbM_t \vec{o}_t(1) \label{fact2},
		\end{align} where $\tilbA$ is an independent copy of $\bA$, and the matrices $\tilbQ_t$ and $\tilbM_t$ are such their columns form orthogonal bases for $\mbox{span}(\boldm^{(0)},\cdots, \boldm^{(t-1)})$ and $\mbox{span}(\bq^{(0)},\cdots, \bq^{(t-1)})$, respectively. 
		
		In addition, let 
		\begin{align}
		\bM_t&:=\big[\boldm^{(0)}|\boldm^{(1)}|\cdots| \boldm^{(t-1)}\big],\\
		\bQ_t&:=\big[\bq^{(0)}|\bq^{(1)}|\cdots| \bq^{(t-1)}\big],\\
		\bX_t&:=\bA^*\bM_t\\
		\bY_t&:=\bA \bQ_t.
		\end{align} Here, we denote by $[\ba_1|\ba_2|\cdots\ba_k]$ the matrix with columns $\ba_1,\ba_2,\cdots,\ba_k$. Then, from Lemma \citep[Lemma 10]{BayatiMonta2011}, it can show that
		\begin{align}
		\bh^{(t+1)}\big| \calF_{t+1,t}^{(n)}\cup \sigma(\bs_0)&\stackrel{(d)}{=} \bH_t(\bM_t^*\bM_t)^{-1}\bM_t^*\boldm_{\|}^{(t)}+ \bP_{\bQ_{t+1}}^{\perp} \tilbA^* \bP_{\bM_t}^{\perp}\boldm^{(t)}+\bQ_t \vec{o}_t(1) \label{lady1},\\
		\bb^{(t)}\big| \calF_{t,t}^{(n)}\cup \sigma(\bs_0)&\stackrel{(d)}{=} \bB_t(\bQ_t^*\bQ_t)^{-1}\bQ_t^*\bq_{\|}^{(t)}+ \bP_{\bM_t}^{\perp} \tilbA \bP_{\bQ_t}^{\perp}\bq^{(t)}+\bM_t \vec{o}_t(1) \label{lady2},
		\end{align} where
		\begin{align}
		\bB_t&:=\big[\bb^{(0)}|\bb^{(1)}|\cdots| \bb^{(t-1)}\big],\\
		\bH_t&:=\big[\bh^{(0)}|\bh^{(1)}|\cdots| \bh^{(t-1)}\big],\\
		\bP_{\bQ_t}^{\perp}&:=\bI-\bP_{\bQ_t},\\
		\bP_{\bM_t}^{\perp}&:=\bI-\bP_{\bM_t},
		\end{align} and $\bP_{\bM_t}$ and $\bP_{\bQ_t}$ are orthogonal projectors onto column spaces of $\bQ_t$ and $\bM_t$, respectively. 
		\item \emph{Step 2:} By using Lemma \ref{GSLLN}, for all pseudo-Lipschitz $\phi_h: \bbR^{t+2}\to \bbR_+$ of order $k$, it holds almost surely that
		\begin{align}
		\lim_{n \to \infty} \frac{1}{n^{\alpha}} \bigg(\sum_{i=1}^n \phi_h(h_i^{(1)},h_i^{(2)},\cdots, h_i^{(t+1)},s_{0,i})-\bbE\bigg[\phi_h(h_i^{(1)},h_i^{(2)},\cdots, h_i^{(t+1)},s_{0,i})\bigg] \bigg)=0
		\label{cur1}.
		\end{align}
		To show \eqref{cur1}, we set
		\begin{align}
		T_n:=\sum_{i=1}^n \phi_h(h_i^{(1)},h_i^{(2)},\cdots, h_i^{(t+1)},s_{0,i}) \label{defSn}.
		\end{align}
		As \citep{BayatiMonta2011}, given $\calF_{t+2,t+1}$, the effects of all terms containing $\vec{o}_{t+1}$ in the distribution of $\bh^{(t+1)}$ and $\bb^{(t+1)}$ in \eqref{fact1}, \eqref{fact2} can be neglected. This fact can be easily explained as follows. Since the limits of $\bh^{(t+1)}$ and $\bb^{(t+1)}$ are Gaussian vectors which have bounded moments, by applying Lemma \ref{lem:lipschitz} and the dominated convergence theorem \citep{Billingsley}, the orders of limits and expectations are interchangeable. Hence, we only need to work with $\phi_h(\lim_{n\to \infty} \bh^{(t+1)})$ and $\phi_b(\lim_{n\to \infty} \bb^{(t+1)})$ instead of $\phi_h(\bh^{(t+1)})$ or $\phi_b(\bh^{(t+1)})$ inside all the expectations of these random variables. 
		
		In all the proofs in this paper, as \cite{BayatiMonta2011}, we also define $\bbE[f(\calF,X)|\calF]$ as the expectation of the random function $f(\calF,X)$ given that all the random in the sigma-algebra $\calF$ are fixed. This also means that $\bbE[f(\calF,X)|\calF]$ is a constant, not a random variable as in the standard definition of the conditional expectation in probability \citep{Billingsley}. 
		
		Now, let
		\begin{align}
		\calF_{t+2,t+1}&:=\calF_{t+2,t+1}^{\infty}\\
		&=\bigcup_{n=1}^{\infty} \calF_{t+2,t+1}^{(n)}.
		\end{align}
		
		Then, for any $\nu>0$ and $\gamma>0$, we have
		\begingroup
		\allowdisplaybreaks
		\begin{align}
		&\bbE\bigg[\max_{1\leq l\leq n} |T_l-\bbE[T_l|\calF_{t+2,t+1} ]|^{2-\nu} \bigg|\calF_{t+2,t+1}\bigg]\nn\\
		&\qquad =\int_0^{\infty}\bbP\big[\max_{1\leq l\leq n}|T_l-\bbE[T_l|\calF_{t+2,t+1}]|^{2-\nu} >t\bigg|\calF_{t+2,t+1}\big]dt\\
		&\qquad=\int_0^{n^{\gamma}} \bbP\big[\max_{1\leq l\leq n}|T_l-\bbE[T_l|\calF_{t+2,t+1}]|^{2-\nu} >t\bigg|\calF_{t+2,t+1}\big]dt\nn\\
		&\qquad \qquad + \int_{n^{\gamma}}^{\infty} \bbP\big[\max_{1\leq l\leq n}|T_l-\bbE[T_l|\calF_{t+2,t+1}]|^{2-\nu} >t\big|\calF_{t+2,t+1}\big]dt\\
		&\qquad \leq n^{\gamma} + \int_{n^{\gamma}}^{\infty} \bbP\big[\max_{1\leq l\leq n}|T_l-\bbE[T_l|\calF_{t+2,t+1}]|^2>t^{\frac{2}{2-\nu}}\bigg|\calF_{t+2,t+1}\big]dt\\
		&\qquad = n^{\gamma}+ \int_{n^{\gamma}}^{\infty} \bbP\big[\max_{1\leq l\leq n}|T_l-\bbE[T_l|\calF_{t+2,t+1}]|>t^{\frac{1}{2-\nu}}\bigg|\calF_{t+2,t+1} \big]dt\\
		&\qquad \leq n^{\gamma}+ \int_{n^{\gamma}}^{\infty} \frac{\var(T_n|\calF_{t+2,t+1})}{t^{\frac{2}{2-\nu}}} dt  \label{komo}\\
		&\qquad= n^{\gamma}+  \bbE(T_n|\calF_{t+2,t+1})\int_{n^{\gamma}}^{\infty}\frac{1}{t^{\frac{2}{2-\nu}}} dt\\
		&\qquad= n^{\gamma}+ o(\var(T_n|\calF_{t+2,t+1})) \label{cute1},
		\end{align}
		\endgroup where \eqref{komo} follows from Kolmogorov's maximal inequality \citep{Billingsley} since given $\calF_{t+2,t+1}$, $T_n-\bbE[T_n|\calF_{t+2,t+1}]$ is a sum of independent random variables with zero means for each $n\geq 1$ by the i.i.d. generation of the sequence $\bs_0=(s_{0,1},s_{0,2},\cdots)$.
		
		It follows that
		\begin{align}
		\var(T_n|\calF_{t+2,t+1})&=\bbE\bigg[ \big(T_n-\bbE[T_n|\calF_{t+2,t+1} ]\big)^2\bigg|\calF_{t+2,t+1} \bigg]\\
		&= \sum_{i=1}^n \var\bigg[\phi_h(h_i^{(1)},h_i^{(2)},\cdots, h_i^{(t+1)},s_{0,i}) \bigg|\calF_{t+2,t+1}\bigg] \label{taf1} \\
		&=\sum_{i=1}^n \bbE\bigg[\phi_h^2(h_i^{(1)},h_i^{(2)},\cdots, h_i^{(t+1)},s_{0,i})  \bigg|\calF_{t+2,t+1}\bigg]\nn\\
		& \qquad -\bigg(\bbE\big[\phi_h(h_i^{(1)},h_i^{(2)},\cdots, h_i^{(t+1)},s_{0,i})  \bigg|\calF_{t+2,t+1}\big]\bigg)^2 \label{paty1},
		\end{align} where \eqref{taf1} follows from the fact that given $\calF_{t+2,t+1}$, $T_n-\bbE[T_n|\calF_{t+2,t+1}]$ is a sum of independent random variables with zero means for each $n\geq 1$ by the i.i.d. generation of the sequence $\bs_0=(s_{0,1},s_{0,2},\cdots)$.

		Now, we have
		\begin{align}
		&\bbE\bigg[\phi_h^2(h_i^{(1)},h_i^{(2)},\cdots, h_i^{(t+1)},s_{0,i})  \bigg|\calF_{t+2,t+1}\bigg]\nn\\
		& =\bbE\bigg[\phi_h^2(h_i^{(1)},h_i^{(2)},\cdots, h_i^{(t+1)},0)  \bigg|\calF_{t+2,t+1}, s_{0,i}=0\bigg]\bbP\bigg[s_{0,i}=0\bigg|\calF_{t+2,t+1}\bigg]\nn\\
		&\qquad + \sum_{b=1}^B \bbE\bigg[\phi_h^2(h_i^{(1)},h_i^{(2)},\cdots, h_i^{(t+1)},a_b)  \bigg|\calF_{t+2,t+1}, s_{0,i}=a_b\bigg]\bbP\bigg[s_{0,i}=a_b\bigg|\calF_{t+2,t+1}\bigg] \label{Y2}.
		\end{align} 
		
		On the other hand, we also have
		\begingroup
		\allowdisplaybreaks
		\begin{align}
		&\bigg(\bbE_{s_{0,i}}\big[\phi_h(h_i^{(1)},h_i^{(2)},\cdots, h_i^{(t+1)},s_{0,i})  \bigg|\calF_{t+2,t+1}\big]\bigg)^2\nn\\
		& =\bigg(\bbE\bigg[\phi_h(h_i^{(1)},h_i^{(2)},\cdots, h_i^{(t+1)},0)  \bigg|\calF_{t+2,t+1}, s_{0,i}=0\bigg]\bbP\bigg[s_{0,i}=0\bigg|\calF_{t+2,t+1}\bigg]\nn\\
		&\qquad +\sum_{b=1}^B \bbE\bigg[\phi_h(h_i^{(1)},h_i^{(2)},\cdots, h_i^{(t+1)},a_b)  \bigg|\calF_{t+2,t+1}, s_{0,i}= a_b\bigg]\bbP\bigg[s_{0,i}=a_b \bigg|\calF_{t+2,t+1}\bigg]\bigg)^2\\
		&\geq \bigg(\bbE\bigg[\phi_h(h_i^{(1)},h_i^{(2)},\cdots, h_i^{(t+1)},0)  \bigg|\calF_{t+2,t+1}, s_{0,i}=0\bigg]\bbP\bigg[s_{0,i}=0\bigg|\calF_{t+2,t+1}\bigg]\bigg)^2\nn\\
		&\qquad +\sum_{b=1}^B \bigg(\bbE\bigg[\phi_h(h_i^{(1)},h_i^{(2)},\cdots, h_i^{(t+1)},a_b)  \bigg|\calF_{t+2,t+1}, s_{0,i}=a_b \bigg]\bbP\bigg[s_{0,i}= a_b\bigg|\calF_{t+2,t+1}\bigg]\bigg)^2\label{awa1}\\
		& = \bigg(\bbE\bigg[\phi_h^2(h_i^{(1)},h_i^{(2)},\cdots, h_i^{(t+1)},0)  \bigg|\calF_{t+2,t+1}, s_{0,i}=0\bigg]\bigg)\bigg(\bbP\bigg[s_{0,i}=0\bigg|\calF_{t+2,t+1}\bigg]\bigg)^2\nn\\
		&+\sum_{b=1}^B \bigg(\bbE\bigg[\phi_h^2(h_i^{(1)},h_i^{(2)},\cdots, h_i^{(t+1)},s_{i,0})  \bigg|\calF_{t+2,t+1}, s_{0,i}=a_b \bigg]\bigg)\bigg(\bbP\bigg[s_{0,i}=a_b \bigg|\calF_{t+2,t+1}\bigg]\bigg)^2 \label{Y3}
		\end{align}
		\endgroup
		where \eqref{awa1} follows from $\phi_h: \bbR^{t+2}\to \bbR_+$ and $(x_1+x_2+\cdots+x_B)^2 \geq \sum_{i=1}^B x_i^2$ if $x_i\geq 0$ for all $i \in [B]$, and \eqref{Y3} follows from the fact that given $\calF_{t+2,t+1}$ and $s_{i,0}$, $\phi_h(h_i^{(1)},h_i^{(2)},\cdots, h_i^{(t+1)},s_{i,0})$ are constants.
		
		From \eqref{paty1}, \eqref{Y2}, and \eqref{Y3}, we obtain
		\begingroup
		\allowdisplaybreaks
		\begin{align}
		&\bbE\bigg[ \big(T_n-\bbE[T_n\big|\calF_{t+2,t+1}]\big)^2\bigg|\calF_{t+2,t+1} \bigg]\nn\\
		& =\sum_{i=1}^n \bigg(\bbE\bigg[\phi_h^2(h_i^{(1)},h_i^{(2)},\cdots, h_i^{(t+1)},0)  \bigg|\calF_{t+2,t+1}, s_{0,i}=0\bigg]\bbP\bigg[s_{0,i}= 0\bigg|\calF_{t+2,t+1}\bigg]\bigg)\nn\\
		& \qquad \qquad \times \bbP\bigg[s_{0,i}\neq 0\bigg|\calF_{t+2,t+1}\bigg]\nn\\
		& +\sum_{i=1}^n \sum_{b=1}^B \bigg(\bbE\bigg[\phi_h^2(h_i^{(1)},h_i^{(2)},\cdots, h_i^{(t+1)},a_b)  \bigg|\calF_{t+2,t+1}, s_{0,i}=a_b\bigg] \bbP\bigg[s_{0,i}\neq a_b \bigg|\calF_{t+2,t+1}\bigg]\bigg)\nn\\
		& \qquad \qquad \times \bbP\bigg[s_{0,i}= a_b \bigg|\calF_{t+2,t+1}\bigg] \label{bug1}\\
		& =\sum_{i=1}^n \bigg(\bbE\bigg[\phi_h^2(h_i^{(1)},h_i^{(2)},\cdots, h_i^{(t+1)},0)  \bigg|\calF_{t+2,t+1}, s_{0,i}=0\bigg]\bbP\bigg[s_{0,i}= 0\bigg|\calF_{t+2,t+1}\bigg]\bigg)\nn\\
		& \qquad \qquad \times \bbP\bigg[s_{0,i}\neq 0\bigg|\calF_{t+2,t+1}\bigg]\nn\\
		&  +\sum_{i=1}^n \sum_{b=1}^B \bigg(\bbE\bigg[\phi_h^2(h_i^{(1)},h_i^{(2)},\cdots, h_i^{(t+1)},a_b)  \bigg|\calF_{t+2,t+1}\bigg] \bbP\bigg[s_{0,i}\neq a_b \bigg|\calF_{t+2,t+1},s_{0,i}\neq a_b\bigg]\bigg)\nn\\
		& \qquad \qquad \times \bbP\bigg[s_{0,i}= a_b \bigg|\calF_{t+2,t+1}\bigg] \label{bug2}\\
		& \leq \sum_{i=1}^n \bbE\bigg[\phi_h^2(h_i^{(1)},h_i^{(2)},\cdots, h_i^{(t+1)},0)  \bigg| s_{0,i}=0\bigg]\bbP\bigg[s_{0,i}\neq 0\bigg|\calF_{t+2,t+1}\bigg]\nn\\
		&+ \sum_{i=1}^n \sum_{b=1}^B \bbE\bigg[\phi_h^2(h_i^{(1)},h_i^{(2)},\cdots, h_i^{(t+1)},s_{0,i})  \bigg|s_{0,i}\neq a_b \bigg]\bbP\bigg[s_{0,i}=a_b\bigg|\calF_{t+2,t+1}\bigg] \label{fix1}\\
		&\leq \sum_{i=1}^n\bbE\bigg[\phi_h^2(h_i^{(1)},h_i^{(2)},\cdots, h_i^{(t+1)},0)  \bigg|\calF_{t+2,t+1}\bigg]\bbP\bigg[s_{0,i}\neq 0\bigg|\calF_{t+2,t+1}\bigg]  \nn\\
		&  + \sum_{i=1}^n\sum_{b=1}^B \bbE\bigg[\phi_h^2(h_i^{(1)},h_i^{(2)},\cdots, h_i^{(t+1)},a_b)  \bigg| \calF_{t+2,t+1}\bigg]\bbP\bigg[s_{0,i}=a_b\bigg|\calF_{t+2,t+1}\bigg]\label{fix2},
		\end{align}
		\endgroup 
		where \eqref{bug2} follows from the fact that given $\calF_{t+2,t+1}$, $\phi_h^2(h_i^{(1)},h_i^{(2)},\cdots, h_i^{(t+1)},a_b) $ does not depends $s_{0,i}=a_b$ or $s_{0,i}\neq a_b$, and \eqref{fix1} follows from $\bbE[Y|B]\bbP(B)\leq \bbE[Y|B]\bbP(B)+ \bbE[Y|B^c]\bbP(B^c)=\bbE[Y]$ if $Y\geq 0$ a.s.
		
		From \eqref{fix1} and \eqref{fix2}, we obtain
		\begingroup
		\allowdisplaybreaks
		\begin{align}
		&\bbE\bigg[\bbE\bigg[ \big(T_n-\bbE[T_n|\calF_{t+2,t+1}]\big)^2\big|\calF_{t+2,t+1} \bigg]\bigg]\nn\\
		&\qquad \leq  \sum_{i=1}^n\bbE\bigg[\bbE\bigg[\phi_h^2(h_i^{(1)},h_i^{(2)},\cdots, h_i^{(t+1)},0)  \bigg|\calF_{t+2,t+1}\bigg]\bbP\bigg[s_{0,i}\neq 0\bigg| \calF_{t+2,t+1}\bigg] \bigg] \nn\\
		&\qquad \qquad + \sum_{i=1}^n \sum_{b=1}^B  \bbE\bigg[\bbE\bigg[\phi_h^2(h_i^{(1)},h_i^{(2)},\cdots, h_i^{(t+1)},a_b)  \bigg|\calF_{t+2,t+1}\bigg]\bbP\bigg[s_{0,i}\neq 0 \bigg| \calF_{t+2,t+1}\bigg]\bigg]
		\label{fix3}.
		\end{align}
		\endgroup
		
		Now, denote by $a_0:=0$. Then, by \citep{BayatiMonta2011}, given $b \in [B]\cup\{0\}$ and $p>1$, it holds that
		\begin{align}
		\frac{1}{n}\sum_{i=1}^n \bbE\bigg[\phi_h^{2p} (h_i^{(1)},h_i^{(2)},\cdots, h_i^{(t+1)},a_b) \bigg]<E_{p,b,t},
		\end{align} for some $E_{p,b,t}<\infty$. Since $B<\infty$, it holds that
		\begin{align}
		\frac{1}{n}\sum_{i=1}^n \bbE\bigg[\phi_h^{2p} (h_i^{(1)},h_i^{(2)},\cdots, h_i^{(t+1)},a_b) \bigg]<E_{p,t} \label{cusum1}, 
		\end{align} where $E_{p,t}:=\max_{b \in [B]} E_{p,b,t}$.
		
		Hence, for any $p>1$ and $q>1$ such that $\frac{1}{p}+\frac{1}{q}=1$, we have
		\begingroup
		\allowdisplaybreaks
		\begin{align}
		&\bbE\bigg[\bbE\bigg[\phi_h^2(h_i^{(1)},h_i^{(2)},\cdots, h_i^{(t+1)},a_b)  \bigg|\calF_{t+2,t+1}\bigg]\bbP\bigg[s_{0,i}\neq 0 \bigg| \calF_{t+2,t+1}\bigg]\bigg]\nn\\
		&\qquad \leq \bigg(\bbE\bigg[\bigg(\bbE\bigg[\phi_h^2(h_i^{(1)},h_i^{(2)},\cdots, h_i^{(t+1)},a_b)  \bigg|\calF_{t+2,t+1}\bigg]\bigg)^p\bigg]\bigg)^{\frac{1}{p}}\nn\\
		&\qquad \qquad \times \bigg(\bbE\bigg[\bigg(\bbP\bigg[s_{0,i}\neq 0 \bigg| \calF_{t+2,t+1}\bigg]\bigg)^q\bigg]\bigg)^{\frac{1}{q}} \label{holder}\\
		&\qquad = \bigg(\bbE\bigg[\bbE\bigg[\phi_h^{2p}(h_i^{(1)},h_i^{(2)},\cdots, h_i^{(t+1)},a_b)  \bigg|\calF_{t+2,t+1}\bigg]\bigg]\bigg)^{\frac{1}{p}}\nn\\
		&\qquad \qquad \times \bigg(\bbE\bigg[\bigg(\bbP\bigg[s_{0,i}\neq 0 \bigg| \calF_{t+2,t+1}\bigg]\bigg)^q\bigg]\bigg)^{\frac{1}{q}} \label{tact}\\
		&\qquad \leq  \bigg(\bbE\bigg[\phi_h^{2p}(h_i^{(1)},h_i^{(2)},\cdots, h_i^{(t+1)},a_b)\bigg]\bigg)^{\frac{1}{p}} \bigg(\bbE\bigg[\bbP\bigg[s_{0,i}\neq 0 \bigg| \calF_{t+2,t+1}\bigg]\bigg]\bigg)^{\frac{1}{q}} \label{tact1}\\
		&\qquad= \bigg(\bbE\bigg[\phi_h^{2p}(h_i^{(1)},h_i^{(2)},\cdots, h_i^{(t+1)},a_b)\bigg]\bigg)^{\frac{1}{p}} \bigg(\bbP\big[s_{0,i}\neq 0\big]\bigg)^{\frac{1}{q}} \\
		&\qquad= \bigg(\frac{k}{n}\bigg)^{\frac{1}{q}} \bigg(\bbE\bigg[\phi_h^{2p}(h_i^{(1)},h_i^{(2)},\cdots, h_i^{(t+1)},a_b)\bigg]\bigg)^{\frac{1}{p}}\\
		&\qquad\leq L n^{(\alpha-1)/q} \bigg(\bbE\bigg[\phi_h^{2p}(h_i^{(1)},h_i^{(2)},\cdots, h_i^{(t+1)},a_b)\bigg]\bigg)^{\frac{1}{p}} \label{small0} \\
		&\qquad \leq L n^{(\alpha-1)/q} E_{p,t}^{1/p} \label{tfact4}
		\end{align}
		\endgroup
		for some $0<L<\infty$, where \eqref{holder} follows from H\"{o}lder's inequality \citep{Royden}, \eqref{tact} follows from given $\calF_{t+2,t+1}$, $\phi_h^2(h_i^{(1)},h_i^{(2)},\cdots, h_i^{(t+1)},a_b)$ is a constant, \eqref{tact1} follows from $q>1$, \eqref{small0} follows from $k=O(n^{\alpha})$, and \eqref{tfact4} follows from \eqref{cusum1}.
		
		From \eqref{cute1} and \eqref{tfact4}, for any $\nu>0$ and $\gamma>0$, we have
		\begin{align}
		&\bbE\bigg[\bbE\bigg[\max_{1\leq l\leq n} |T_l-\bbE[T_l|\calF_{t+2,t+1} ]|^{2-\nu} \bigg|\calF_{t+2,t+1}\bigg]\bigg]\nn\\
		&\qquad \leq  n^{\gamma}+ (B+1)L n n^{(\alpha-1)/q} E_{p,t}^{1/p}
		\label{astar}. 
		\end{align} 
		
		Then, by setting $b_l=l, d_l=l^{\alpha-1},\nu=1/2, r=2-\nu,\rho=\frac{1}{2}\bone\{\alpha\geq 1/2\}+ \frac{\alpha}{2(2-3\alpha)}\bone \{\alpha<\frac{1}{2}\}, q=1+\rho$ and $\nu_l=1+L(B+1) E_{p,t}^{1/p}$ for all $l \in \bbZ^+$, we have
		\begin{align}
		\gamma&:=1+ \frac{\alpha-1}{q}\\
		&\geq 1+ \alpha-1 \\
		&=\alpha>0.
		\end{align}
		In addition, by these settings, from Lemma \ref{GSLLN}, we also have
		\begin{align}
		\sum_{l=1}^{\infty} \frac{\nu_l}{b_l^r d_l^{r-1/(1+\rho)}}&=
		((B+1)L E_{p,t}^{1/p}+1)\sum_{l=1}^{\infty}\frac{1}{l^{2-\nu} l^{(\alpha-1)(2-\nu-1/(1+\rho))}}\\
		&= ((B+1)L E_{p,t}^{1/p}+1)\sum_{l=1}^{\infty}\frac{1}{l^{(2-\nu)\alpha+ (1-\alpha)/(1+\rho)}}\\
		&<\infty \label{cufact},
		\end{align} where \eqref{cufact} follows from $(2-\nu)\alpha+ (1-\alpha)/(1+\rho)>1$ with the set value of $\rho$ (note that $0< \alpha\leq 1$). In addition, $\{b_l=l\}$ is a non-decreasing sequence, $\{d_l=l^{\alpha-1}\}$ is a non-increasing sequence with $d_1=1$, and $b_l d_l=l^{\alpha} \to \infty$ as $l\to \infty$. Hence, by applying Lemma \ref{GSLLN}, given $\calF_{t+2,t+1}$, we have
		\begin{align}
		\frac{1}{n^{\alpha}} \big( T_n-\bbE[T_n\big|\calF_{t+2,t+1}]\big) \to 0, \qquad a.s. \label{gh}
		\end{align} Hence, we obtain \eqref{cur1} from \eqref{gh} and \eqref{defSn}.
		
		Similarly, by using Lemma \ref{GSLLN}, for all pseudo-Lipschitz $\phi_b: \bbR^{t+2}\to \bbR_+$ of order $k$, it holds almost surely that
		\begin{align}
		\lim_{m\to \infty} \frac{1}{m}\bigg(\sum_{i=1}^m \phi_b(b_i^{(1)},b_i^{(2)},\cdots,b_i^{(t+1)},\tilw_i)-\bbE\bigg[\phi_b(b_i^{(1)},b_i^{(2)},\cdots,b_i^{(t+1)},\tilw_i)\bigg]\bigg)=0 \label{batcha1}.
		\end{align}
		To show \eqref{batcha1}, we set
		\begin{align}
		\tilT_m:=\sum_{i=1}^m \phi_b(b_i^{(1)},b_i^{(2)},\cdots,b_i^{(t+1)},\tilw_i) \label{deftilSm}.
		\end{align}
		Then, by \citep{BayatiMonta2011}, it holds that
		\begin{align}
		\frac{1}{m}\sum_{i=1}^m \bbE\bigg[\phi_b^2(b_i^{(1)},b_i^{(2)},\cdots,b_i^{(t+1)},\tilw_{0,i})\bigg]<\tilE_{1,t}< \infty.
		\end{align}
		Then, by setting $b_l=l, d_l=1, r=2, \rho=0$, and $\nu_l=\tilE_{1,t}$ for all $l\in \bbZ^+$ in Lemma \ref{GSLLN}, it follows that
		\begin{align}
		\sum_{l=1}^{\infty}\frac{\nu_l}{b_l^r d_l^{r-1/(1+\rho)}}&=\tilE_{1,t}\sum_{l=1}^{\infty}\frac{1}{l^2}\\
		&<\infty.
		\end{align}
		Similar to the proof of \eqref{cur1}, by applying Lemma \ref{GSLLN}, given $\calF_{t+2,t+1}$, we have
		\begin{align}
		\frac{1}{m}\big(\tilT_m-\bbE[\tilT_m\big|\calF_{t+2,t+1}]\big)\to 0, \quad a.s. \label{tk2}
		\end{align}
		Hence, we obtain \eqref{batcha1} from \eqref{deftilSm} and \eqref{tk2}.
		\item \emph{Step 3:} From \citep[Lemma 2]{BayatiMonta2011}, it holds that $\big[\tilbA^* \boldm_{\perp}^{(t)}]_i \sim \calN(0, \frac{1}{m} \|\boldm_{\perp}^{(t)} \|^2)$. Hence, from \eqref{fact1}, we have
		\begin{align}
		&\bbE\bigg[\phi_h(h_i^{(1)},h_i^{(2)},\cdots, h_i^{(t+1)},s_{0,i})\bigg|\calF_{t+2,t+1}\cup \sigma(\bs_0)\bigg]\\
		&\qquad =\bbE\bigg[\phi_h\bigg(h_i^{(1)},h_i^{(2)},\cdots, \sum_{r=0}^{t-1}\alpha_r h_i^{(r+1)}+\frac{\|\boldm_{\perp}^{(t)} \|Z}{\sqrt{m}},s_{0,i}\bigg)\bigg|\calF_{t+2,t+1}\cup \sigma(\bs_0)\bigg]
		\end{align}  where $Z \sim \calN(0,1)$. It follows that
		\begin{align}
		& \bbE\bigg[\phi_h(h_i^{(1)},h_i^{(2)},\cdots, h_i^{(t+1)},s_{0,i})\bigg]\nn\\
		&\qquad =\bbE\bigg[\bbE\bigg[\phi_h(h_i^{(1)},h_i^{(2)},\cdots, h_i^{(t+1)},s_{0,i})\bigg|\calF_{t+2,t+1}\cup \sigma(\bs_0)\bigg]\bigg] \label{h1}\\
		&\qquad =\bbE\bigg[\bbE\bigg[\phi_h\bigg(h_i^{(1)},h_i^{(2)},\cdots, \sum_{r=0}^{t-1}\alpha_r h_i^{(r+1)}+\frac{\|\boldm_{\perp}^{(t)} \|Z}{\sqrt{m}},s_{0,i}\bigg)\bigg|\calF_{t+2,t+1}\cup \sigma(\bs_0)\bigg]\bigg]\\
		&\qquad=\bbE\bigg[\phi_h\bigg(h_i^{(1)},h_i^{(2)},\cdots, \sum_{r=0}^{t-1}\alpha_r h_i^{(r+1)}+\frac{\|\boldm_{\perp}^{(t)} \|Z}{\sqrt{m}},s_{0,i}\bigg)\bigg], \label{h2}
		\end{align} where \eqref{h1} and \eqref{h2} follow from the tower property of the conditional expectation \citep{Durrett}.
		
		Hence, from \eqref{cur1} and \eqref{h2}, we obtain
		\begin{align}
		&\lim_{n \to \infty} \frac{1}{n^{\alpha}} \bigg(\sum_{i=1}^n \phi_h(h_i^{(1)},h_i^{(2)},\cdots, h_i^{(t+1)},s_{0,i})\nn\\
		&\qquad \qquad -\bbE\bigg[\phi_h\bigg(h_i^{(1)},h_i^{(2)},\cdots, \sum_{r=0}^{t-1}\alpha_r h_i^{(r+1)}+\frac{\|\boldm_{\perp}^{(t)} \|Z}{\sqrt{m}},s_{0,i}\bigg)\bigg] \bigg)=0 \label{cur1b}.
		\end{align}
		Similarly, from \eqref{fact2} and the tower property of the conditional expectation, we can show that
		\begin{align}
		\bbE\bigg[\phi_b(b_i^{(1)},b_i^{(2)},\cdots, b_i^{(t+1)},\tilw_i)\bigg]=\bbE\bigg[\phi_b\bigg(b_i^{(1)},b_i^{(2)},\cdots, \sum_{r=0}^{t-1}\beta_r b_i^{(r)}+\frac{\|\bq_{\perp}^{(t)} \|Z}{\sqrt{m}},\tilw_i\bigg)\bigg] \label{h3}.
		\end{align}
		From \eqref{batcha1} and \eqref{h3}, we obtain
		\begin{align}
		&\lim_{m\to \infty} \frac{1}{m}\bigg(\sum_{i=1}^m \phi_b(b_i^{(1)},b_i^{(2)},\cdots,b_i^{(t+1)},\tilw_i)\nn\\
		&\qquad-\bbE\bigg[\phi_b\bigg(b_i^{(1)},b_i^{(2)},\cdots, \sum_{r=0}^{t-1}\beta_r b_i^{(r)}+\frac{\|\bq_{\perp}^{(t)} \|Z}{\sqrt{m}},\tilw_i\bigg)\bigg]\bigg)=0 \label{batcha2}.
		\end{align}
		By using induction, from \eqref{cur1b} and \eqref{batcha2}, we have
		\begin{align}
		&\lim_{n\to \infty} \frac{1}{n^{\alpha}}\bigg(\sum_{i=1}^n \phi_h(h_i^{(1)},h_i^{(2)},\cdots, h_i^{(t+1)},s_{0,i})\nn\\
		&\qquad -\bbE\bigg[\phi_h(\ttau_0 Z_0,\ttau_1 Z_1,\cdots, \ttau_t Z_t, U_0)\bigg]\bigg)=0 \label{test2}\\
		&\lim_{m\to \infty} \frac{1}{m}\bigg(\sum_{i=1}^m \phi_b(b_i^{(1)},b_i^{(2)},\cdots, b_i^{(t+1)},\tilw_i)\nn\\
		&\qquad -\bbE\bigg[\phi_b(\tsigma_0 \hatZ_0,\tsigma_1 \hatZ_1,\cdots, \tsigma_t \hatZ_t, \tilW)\bigg]\bigg)=0 \label{test3},
		\end{align} where
		\begin{align}
		\ttau_t^2&:=\bbE\bigg[\bigg(\sum_{r=0}^{t-1}\alpha_r \ttau_r Z_r +\frac{\|\boldm_{\perp}^{(t)} \|Z}{\sqrt{m}}\bigg)^2\bigg], \label{tag1}\\
		\tsigma_t^2&:=\bbE\bigg[\bigg(\sum_{r=0}^{t-1}\beta_r \tsigma_r \hatZ_r +\frac{\|\bq_{\perp}^{(t)} \|Z}{\sqrt{m}}\bigg)^2\bigg] \label{tag2},
		\end{align} and $(Z_1,Z_2,\cdots, Z_{t-1},Z_t)$ and $(\hatZ_1,\hatZ_2,\cdots,\hatZ_{t-1},\hatZ_t)$ are two Gaussian vectors independent of $U_0$ and $W$ with $Z_i, \hatZ_i \sim \calN(0,1)$.
		\item \emph{Step 4:} Finally, we show that 
		\begin{align}
		\ttau_t^2 - \tau_t^2 \to 0 \label{beca1},\\
		\tsigma_t^2- \sigma_t^2 \to 0\label{beca2}, 
		\end{align} where $\tau_t$ and $\sigma_t$ follow the state evolutions in \eqref{rectaut} and \eqref{recursigmat}, respectively. 
		
		Indeed, by setting
		\begin{align}
		\phi_h(\nu^{(1)},\nu^{(2)},\cdots, \nu^{(t+1)},u)&:=\big(\nu^{(t+1)}\big)^2,\\ \phi_b(\nu^{(1)},\nu^{(2)},\cdots, \nu^{(t+1)},u)&:=\big(\nu^{(t+1)}\big)^2
		\end{align}
		for all $(\nu^{(1)},\nu^{(2)},\cdots, \nu^{(t+1)},u) \in \bbR^{t+2}$ for all $(\nu^{(1)},\nu^{(2)},\cdots, \nu^{(t+1)},u) \in \bbR^{t+2}$, from \eqref{test2} and \eqref{test3}, we have
		\begin{align}
		\lim_{n\to \infty} \frac{1}{n^{\alpha}} \bigg( \big\|\bh^{(t+1)}\big\|^2-n \bbE\bigg[\bigg(\sum_{r=0}^{t-1}\alpha_r\ttau_r Z_r+\frac{\|\boldm_{\perp}^{(t)} \|Z}{\sqrt{m}}\bigg)^2\bigg]\bigg)&=0, \quad \forall t\geq 1 \label{cur3},\\
		\lim_{n\to \infty} \frac{1}{m} \bigg( \big\|\bb^{(t+1)}\big\|^2-m \bbE\bigg[\bigg(\sum_{r=0}^{t-1}\alpha_r\tsigma_r \hatZ_r+\frac{\|\bq_{\perp}^{(t)} \|Z}{\sqrt{m}}\bigg)^2\bigg]\bigg)&=0, \quad \forall t\geq 1 \label{cur4}.
		\end{align}
		It follows from \eqref{tag1}, \eqref{cur3}  and \eqref{tag2}, \eqref{cur4} that
		\begin{align}
		\lim_{n\to \infty} \frac{1}{n^{\alpha}} \bigg( \big\|\bh^{(t+1)}\big\|^2-n\ttau_t^2\bigg)=0 \label{qfact1},\\
		\lim_{m\to \infty} \frac{1}{m} \bigg( \big\|\bb^{(t+1)}\big\|^2-m\tsigma_t^2\bigg)=0 \label{qfact1b}.
		\end{align}
		On the other hand, from \eqref{lady1} and \eqref{lady2}, we can prove that (cf. \citep[Eq. (3.18) and (3.19)]{BayatiMonta2011})
		\begin{align}
		\lim_{n\to \infty}\bigg( \frac{\|\bh^{(t+1)}\|^2}{n}-\frac{\|\boldm^{(t)}\|^2}{m}\bigg)=0 \label{eq195},\\
		\lim_{n\to \infty}\bigg( \|\bb^{(t+1)}\|^2-\lim_{n\to \infty}\|\bq^{(t)}\|^2\bigg)=0 \label{eq195b},
		\end{align} where $\boldm^{(t)}=g_t(\bb^{(t)},\tilbw)$ and $\bq^{(t)}=f_t(\bh^{(t)},\bs_0)$ as \eqref{defbm} and \eqref{deffbq}, respectively. Hence, from \eqref{qfact1} -- \eqref{eq195b}, we obtain
		\begin{align}
		\lim_{n \to \infty}  \frac{\big\|g_t(\bb^{(t)},\tilbw)\big\|^2}{m}- \ttau_t^2=0 \label{qfact2},\\
		\lim_{m \to \infty} \frac{1}{m}  \big\|f_t(\bh^{(t)},\bs_0)\big\|^2-\tsigma_t^2=0 \label{qfact2b}.
		\end{align}
		Furthermore, by setting 
		\begin{align}
		\phi_b(\nu^{(1)},\nu^{(2)},\cdots, \nu^{(t+1)},u)&:=g_t\big(\nu^{(t+1)},u\big)^2,\\  \phi_h(\nu^{(1)},\nu^{(2)},\cdots, \nu^{(t+1)},u)&:=f_t\big(\nu^{(t+1)},u\big)^2 
		\end{align}
		for all $(\nu^{(1)},\nu^{(2)},\cdots, \nu^{(t+1)},u) \in \bbR^{t+2}$, from \eqref{test3} and \eqref{test2}, we have
		\begin{align}
		\lim_{m\to \infty} \frac{1}{m} \big\|g_t(\bb^{(t)},\tilbw)\big\|^2-\bbE[g_t(\tsigma_tZ,\tilW)^2]&=0 \label{qfact3},\\
		\lim_{n\to \infty} \frac{\big\|f_t(\bh^{(t)},\bs_0)\big\|^2}{n^{\alpha}} -n^{1-\alpha}\bbE[f_t(\ttau_tZ,U_0)^2]&=0 \label{qfact3b}.
		\end{align}
		From \eqref{qfact2}, \eqref{qfact3}, and \eqref{qfact2b}, \eqref{qfact3b}, and $m=\delta n^{\alpha}$, we finally obtain
		\begin{align}
		\lim_{n\to \infty}  \ttau_t^2-\bbE[g_t(\ttau_tZ,\tilW)^2]=0 \label{eq199},\\
		\lim_{n\to \infty} \tsigma_t^2 - \frac{n^{1-\alpha}}{\delta}\bbE[f_t(\ttau_tZ,U_0)^2]&=0 \label{eq200}.
		\end{align} 
		Finally, we achieve \eqref{beca1} and \eqref{beca2} by combining \eqref{eq199}, \eqref{eq200}  and \eqref{rectaut}, \eqref{recursigmat}.
	\end{itemize} 
\end{proof}
Finally, we prove the following theorem.
\noindent
\begin{theorem} \label{thm:par0} Let $\hatbx^{(t)}=(\hatx_1^{(t)},\hatx_2^{(t)},\cdots, \hatx_n^{(t)})$ be the estimate of $\bS$ at the step $t$ in Algorithm 1. Then, for any pseudo-Lipschitz function $\psi: \bbR^2 \to \bbR_+$ of order $k$ and all $t\geq 0$, the following holds almost surely
	\begin{align}
	\lim_{n\to \infty}\bigg( \bigg[\frac{1}{n^{\alpha}}\sum_{i=1}^n\psi(\hatx_i^{(t+1)},S_i)\bigg]-\delta (\tau_t^2-\Delta_n)\bigg)=0 \label{eqtwoh},
	\end{align}
	where $\tau_t$ satisfies the following state evolution:
	\begin{align}
	\tau_0^2&=\Delta_n + \frac{n^{1-\alpha}}{\delta}\bbE_{S \sim \tilP_0}[S^2] \label{eq12},\\
	\tau_{t+1}^2&=\Delta_n +\frac{n^{1-\alpha}}{\delta}\bbE\big[\psi\big(\eta_t(U_0+\tau_t Z),U_0\big)\big] \quad \forall t\in \bbZ_+ \label{eq13},
	\end{align} where $U_0$ and $Z$ are independent, $U_0 \sim \tilP_0$ and $Z \sim \calN(0,1)$.
\end{theorem}
\begin{remark} Some remarks are in order.
\begin{itemize}
\item At $\alpha=1$, Theorem \ref{thm:par0} recovers \citep[Theorem 1]{BayatiMonta2011}.
\item $\tau_t$ depends on $\alpha$ though $\tau_0$ since $\tilP_0$ is a function of $\alpha$.
\end{itemize}
\end{remark}
\begin{proof} 
	First, in Lemma \ref{thm:chang}, let
	\begin{align}
	\tilbw&:=\bw\sqrt{\Delta_n},\\
	\bh^{(t+1)}&:=\bs-(\bA^* \bz^{(t)}+ \hatbx^{(t)}) \label{bes1} \\
	\bq^{(t)}&:=\hatbx^{(t)}-\bs\\
	\bb^{(t)}&:=\tilbw-\bz^{(t)}\\
	\boldm^{(t)}&:=-\bz^{(t)},\\
	\bs_0&:=\bs.
	\end{align}
	In addition, the function $f_t$ and $g_t$ are given by 
	\begin{align}
	f_t(\bu,\bs)&:=\eta_{t-1}(\bs-\bu)-\bs\\ g_t(\bu,\tilbw)&:=\bu-\tilbw \label{defftgt}
   \end{align} and the initial condition $\bq^{(0)}:=-\bs$. 
	Then, we recover Algorithm 1 as a special case. Hence, by defining:
	$
	\phi_h(v^{(1)},v^{(2)},\cdots,v^{(t+1)},s):=\psi(\eta_t(s-v^{(t+1)}),s), \qquad \forall (v^{(1)},v^{(2)},\cdots,v^{(t+1)})\in \bbR^{t+1},
	$ which is a pseudo-Lipschitz function \citep{BayatiMonta2011}, from Lemma \ref{thm:chang}, we have
	\begin{align}
	\lim_{n\to \infty} \frac{1}{n^{\alpha}}\sum_{i=1}^n\psi\big(\eta_t(s_i-h_i^{(t+1)}),s_i\big)-n^{1-\alpha}\bbE\big[\psi(\eta_t(U_0+\tau_t Z_t),U_0)\big]=0, \qquad a.s. \label{bes2}
	\end{align}
	Note that, by Algorithm 1, we have
	\begin{align}
	\hatbx^{(t+1)}=\eta_t\big(\bA^* \bz^{(t)}+ \hatbx^{(t)}\big) \label{bes3}.
	\end{align}
	Hence, from \eqref{bes1}, \eqref{bes2}, and \eqref{bes3}, we obtain
	\begin{align}
	\lim_{n\to \infty} \frac{1}{n^{\alpha}}\sum_{i=1}^n\psi(\hatx_i^{(t+1)},s_i)-n^{1-\alpha}\bbE\big[\psi(\eta_t(U_0+\tau_t Z_t),U_0)\big]=0,
	\end{align}  where $Z_t \sim \calN(0,1)$.
	
	Furthermore, by \eqref{rectaut}, \eqref{recursigmat}, and \eqref{defftgt}, we have
	\begin{align}
	\tau_{t+1}^2&=\Delta_n + \sigma_{t+1}^2\\
	&=\Delta_n + \frac{n^{1-\alpha}}{\delta}\bbE\big[\psi\big(\eta_t(U_0+\tau_t Z),U_0\big) \big]\bigg)\\
	&=\Delta_n + \frac{n^{1-\alpha}}{\delta}\bbE\big[\psi\big(\eta_t(U_0+\tau_t Z),U_0\big) \big] \label{pq1}
	\end{align} with
	\begin{align}
	\tau_0^2&=\Delta_n+\sigma_0^2\\
	&=\Delta_n+ \frac{n^{1-\alpha}}{\delta }\bigg(\frac{\bbE\big[\|\bq^{(0)}\|^2\big]}{n}\bigg)\\
	&=\Delta_n+\frac{n^{1-\alpha}}{\delta }\bigg(\frac{\bbE\big[\|\bS\|^2\big]}{n}\bigg)\\
	&=\Delta_n+\frac{n^{1-\alpha}}{\delta} \bbE_{S \sim \tilP} \bbE[S^2].
	\end{align}
	From \eqref{bes2} and \eqref{pq1}, we obtain \eqref{eqtwoh}. This concludes our proof of Theorem \ref{thm:par0}.
\end{proof}

By setting $\psi(x,y)=(x-y)^2$ and using \eqref{eq13}, the following corollary is easily derived from Theorem \ref{thm:par0}.
\begin{corollary} \label{thm2} With the same notations as Theorem \ref{thm:par0}, the following holds:
	\begin{align}
	\lim_{n\to \infty} \bigg(\bigg[\frac{1}{n^{\alpha}}\sum_{i=1}^n \big(\hatx_i^{(t)}-S_i\big)^2\bigg]-\delta \big(\tau_t^2- \Delta_n \big)\bigg)=0 \label{bapcai},
	\end{align} where $\tau_t$ satisfies the following state evolution:
	\begin{align}
	\tau_0^2&=\Delta_n + \frac{n^{1-\alpha}}{\delta}\bbE_{S \sim \tilP_0}[S^2] \label{eq12cor},\\
	\tau_{t+1}^2&=\Delta_n +\frac{n^{1-\alpha}}{\delta}\bbE\big[\big(\eta_t(U_0+\tau_t Z)-U_0\big)^2\big] \quad \forall t\in \bbZ_+ \label{eq13cor},
	\end{align} where $U_0 \sim \tilP_0$ and $Z \sim \calN(0,1)$.
\end{corollary}
\section{Numerical Evaluations}
In this section, we compare the normalized MMSE fundamental limit in Theorem \ref{mainthm} and the normalized MSE of Algorithm 1 in Corollary \ref{thm2} for the Bernoulli-Rademacher prior. Here, the normalization means that we divide the total mean square error by $k=n^{\alpha}$. 

More specifically, let $\Delta_n=\Delta=s_{\max}<\infty$ and $\tilP_0(s)=(1-\frac{k}{n})\delta(s)+\frac{1}{2}(\frac{k}{n})(\delta(s-\sqrt{\Delta})+\delta(s+\sqrt{\Delta}))$, which is the Bernoulli-Rademacher distribution (cf. \citep{Barbier2020}). With this assumption, we have
\begin{align}
i_{n,\rm{den}}(\Sigma)
&=n^{1-\alpha}I(S;S+\tilW \Sigma)\\
&=n^{1-\alpha}\bigg[H(Y)-\frac{1}{2}\log(2\pi e\Sigma^2)\bigg],
\end{align}	where
\begin{align}
f_Y(y)
&=\bigg(1-\frac{k}{n}\bigg)\frac{1}{\Sigma\sqrt{2\pi}}\exp\bigg(-\frac{y^2}{2\Sigma^2}\bigg)\nn\\
&\quad+\frac{1}{2\Sigma \sqrt{2\pi}}\bigg(\frac{k}{n}\bigg)\bigg(\exp\bigg(-\frac{(y-\sqrt{\Delta})^2}{2\Sigma^2}\bigg) +\exp\bigg(-\frac{(y+\sqrt{\Delta})^2}{2\Sigma^2}\bigg)\bigg).
\end{align}

For this prior distribution, we run Algorithm 1 for $\rm{itermax}=10$ iterations with the denoiser defined as following:
\begin{align}
&\eta(x,\tau)=\bbE[S|S+\tau Z=x]\\
&=\frac{\frac{1}{2}\big(\frac{k}{n}\sqrt{\Delta}\big)\big[\exp(\frac{x\sqrt{\Delta}}{\tau^2})-\exp(-\frac{x\sqrt{\Delta}}{\tau^2})\big]}{\big(1-\frac{k}{n}\big)\exp\big(\frac{\Delta}{2\tau^2}\big)+\frac{1}{2}\big(\frac{k}{n}\big)\big[\exp(\frac{x\sqrt{\Delta}}{\tau^2})+\exp(-\frac{x\sqrt{\Delta}}{\tau^2})\big]}.
\end{align}
 
This denoiser has the following derivative:
\begin{align}
\frac{d\eta(x,\tau)}{dx}&=\frac{\frac{\Delta}{2\tau^2}\big(1-\frac{k}{n}\big)\frac{k}{n}\exp\big(\frac{\Delta}{2\tau^2}\big)\big[\exp(\frac{x\sqrt{\Delta}}{\tau^2})+\exp(-\frac{x\sqrt{\Delta}}{\tau^2})\big]}{\big(\big(1-\frac{k}{n}\big)\exp\big(\frac{\Delta}{2\tau^2}\big)+\frac{1}{2}\big(\frac{k}{n}\big)\big[\exp(\frac{x\sqrt{\Delta}}{\tau^2})+\exp(-\frac{x\sqrt{\Delta}}{\tau^2})\big]\big)^2}\nn\\
&+ \frac{\big(\frac{k}{n}\big)^2\frac{\Delta}{\tau^2}}{\big(\big(1-\frac{k}{n}\big)\exp\big(\frac{\Delta}{2\tau^2}\big)+\frac{1}{2}\big(\frac{k}{n}\big)\big[\exp(\frac{x\sqrt{\Delta}}{\tau^2})+\exp(-\frac{x\sqrt{\Delta}}{\tau^2})\big]\big)^2}.
\end{align}
\begin{figure}[H]
	\centering
	\includegraphics[width=0.8\linewidth]{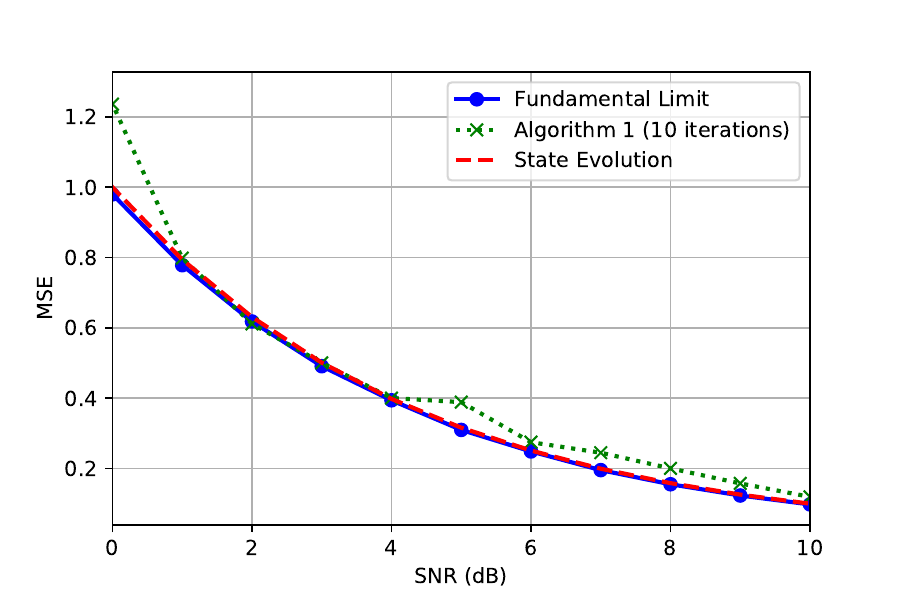}
	\caption{MMSE and the MSE of Algorithm 1 as a function of  SNR at $\alpha=0.5$ and $\delta=0.5$ for $n=1000$. Here, $SNR:=-10\log (\Delta_n/\delta)$ (dB).}
	\label{fig:DMK2}
\end{figure} 
\begin{figure}[H]
	\centering
	\includegraphics[width=0.8\linewidth]{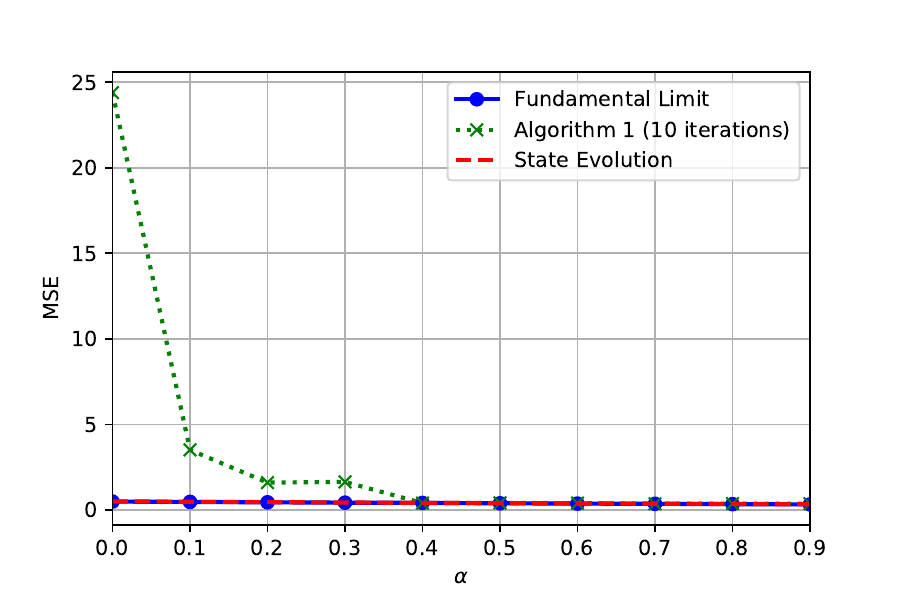}
	\caption{MSE of Algorithm 1, State Evolution, and MSE fundamental limit as functions of $\alpha$ at $\delta=0.5, SNR=10\log (2 \alpha)$ dB for $n=1000$.}
	\label{fig:DMK3}
\end{figure}
\begin{figure}[H]
	\centering
	\includegraphics[width=0.8\linewidth]{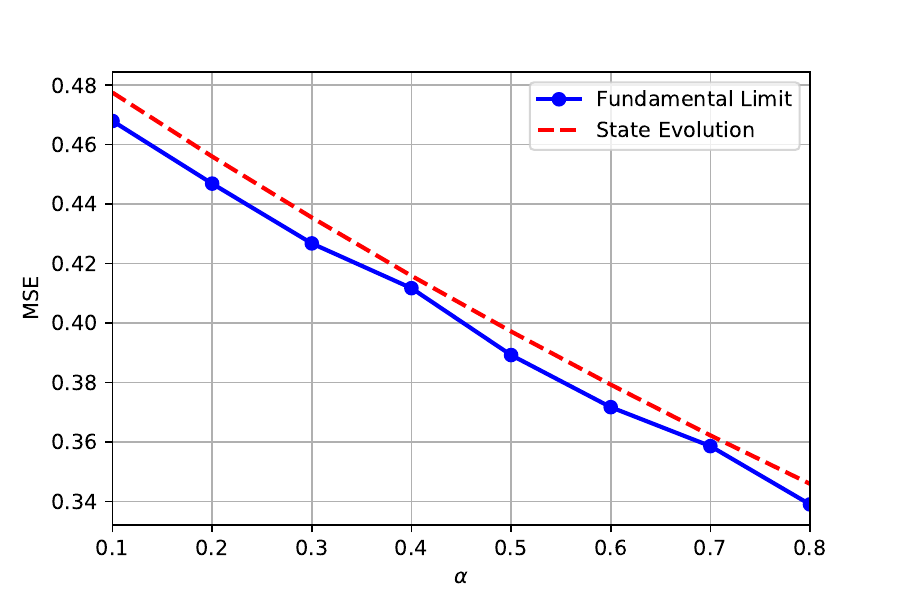}
	\caption{State Evolution of Algorithm 1 vs. Fundamental Limit as functions of $\alpha$ at $\delta=0.5, SNR=10\log (2 \alpha)$ dB for $n=1000$.}
	\label{fig:DMK5}
\end{figure}
\begin{figure}[H]
	\centering
	\includegraphics[width=0.8\linewidth]{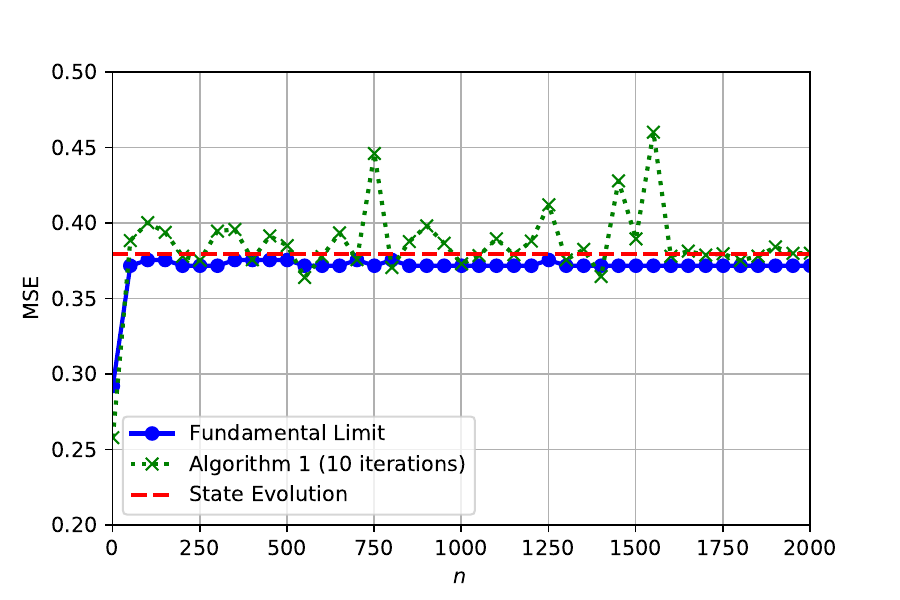}
	\caption{MSE of Algorithm 1, State Evolution, and MSE fundamental limit as functions of $n$ at $\delta=0.5$ and $\alpha=0.6$.}
	\label{fig:DMKL}
\end{figure}

More specifically, let $\Delta_n=\Delta$ and $\tilP_0(s)=(1-\frac{k}{n})\delta(s)+\frac{1}{2}(\frac{k}{n})(\delta(s-1)+\delta(s+1))$, which is the Bernoulli-Rademacher distribution (cf. \citep{Barbier2020}). With this assumption, we have
\begin{align}
i_{n,\rm{den}}(\Sigma)
&=n^{1-\alpha}I(S;S+\tilW \Sigma)\\
&=n^{1-\alpha}\bigg[H(Y)-\frac{1}{2}\log(2\pi e\Sigma^2)\bigg],
\end{align}	where $Y=S+\tilW \Sigma$ and
\begin{align}
f_Y(y)
&=\bigg(1-\frac{k}{n}\bigg)\frac{1}{\Sigma\sqrt{2\pi}}\exp\bigg(-\frac{y^2}{2\Sigma^2}\bigg)\nn\\
&\quad+\frac{1}{2\Sigma \sqrt{2\pi}}\bigg(\frac{k}{n}\bigg)\bigg[\exp\bigg(-\frac{(y-1)^2}{2\Sigma^2}\bigg) +\exp\bigg(-\frac{(y+1)^2}{2\Sigma^2}\bigg)\bigg].
\end{align}

In the first experiment, we set $n=300$ and run AMP in Algorithm 1 for $10$ iterations. Fig.~\ref{fig:DMK2} shows that the MSE achieved by Algorithm 1 via Monte-Carlo simulation is very close to the MMSE fundamental limit in Theorem \ref{mainthm}. The state evolution in Corollary \ref{thm2} tracks the MMSE fundamental limit in Theorem \ref{mainthm} very well. This plot also hints us that a judicious modification of the existing AMPs for linear regimes (for example,~\citep{Donoho2009a}) can work well for sub-linear regimes. 

Fig.~\ref{fig:DMK3} plots the MSE as a function of $\alpha \in (0,1)$ for $\rm{SNR}=10\log (2 \alpha)$ dB for all $\alpha \in [0,1]$. As we can observe from the plot, the gap between the state evolution and fundamental limit is very small. However, there is big gap between the state evolution and MSE from Algorithm 1 at low $\alpha$'s. This can be explained by observing that $m=\delta n^{\alpha}$ is very small at small values of $\alpha$ (for example, $m=1$ at $n=1000$ and $\alpha=0.1$), so the LLNs in Lemma \ref{thm:chang} do not hold. To have a better view of the relationship between the MMSE fundamental limit in Theorem 1 and the state evolution of Algorithm 1, we zoom out it in Fig. \ref{fig:DMK5}.  
 
In Fig.~\ref{fig:DMKL}, we plot MSE as a function of $n$ for fixed $\alpha=0.8$ and $\delta=0.5$. The figure shows that the gap among fundamental limit in Theorem 1, the state evolution of Algorithm 1 in Corollary 14, and MSE of Algorithm 1 nearly coincide to each other at $n$ sufficiently large.





\acks{The author is grateful to Prof.\ Ramji Venkataramanan, the University of Cambridge, for many suggestions to improve the manuscript. The authors also would like to thank the action editor and reviewers for many useful suggestions to improve the manuscript during the review process.}

\vskip 0.2in
\bibliographystyle{unsrt}
\bibliography{isitbib}
\newpage
\appendix
\section*{Appendix A.} \label{beo1proof}
In this Appendix, we provide a proof for Lemma \ref{beo1}.\\

For $\alpha=1$, it is known from \citep[Eq.~(93)]{Barbier2017TheAI} that
\begin{align}
\int_{a_n}^{b_n}\frac{1}{K_n}\sum_{k=1}^{K_n}\int_0^1 dt \frac{d\gamma_k(t)}{dt}\bigg({\rm{ymmse}_{k,t;\eps}} 
-\frac{\rm{mmse}_{k,t;\eps}}{n^{1-\alpha}+\gamma_k(t)\rm{mmse}_{k,t;\eps}}  \bigg) =O(a_n^{-2}n^{-\gamma}) \label{bunhi1a}
\end{align} for some $0<\gamma<1$.

Now, assume that $0 < \alpha<1$\footnote{Our proof of Lemma \ref{beo1} for $\alpha<1$ is simpler than the proof in \citep{Barbier2017TheAI} for $\alpha=1$. More specifically, the proof of concentration inequality in \eqref{eq153} has been simplified by making use of signal sparsity for $\alpha<1$.}. Observe that
\begin{align}
n^{\alpha-1}{\rm{mmse}_{k,t;\eps}}&=\frac{1}{n}\sum_{i=1}^n\bbE\big[(S_i-\langle X_i\rangle_{k,t;\eps})^2\big]\\
&=\frac{1}{n}\sum_{i=1}^n \bbE\big[\big(\langle S_i-X_i \rangle_{k,t;\eps}\big)^2\big]\\
&=\frac{1}{n}\sum_{i=1}^n \bbE\big[\langle \barX_i \rangle_{k,t;\eps}^2\big]
\end{align} since by definition $\barX_i:=X_i-S_i$ for all $ i \in [n]$.

Now, by \citep[Section 6]{BarbierALT2016}, for all $k \in [K_n]$ we have
\begin{align}
{\rm{ymmse}_{k,t;\eps}}=\calY_{1,k}-\calY_{2,k} \label{eq125},
\end{align}
where
\begin{align}
\calY_{1,k}&:=\bbE\bigg[\frac{1}{m}\sum_{\mu=1}^m \big(W_{\mu}^{(k)}\big)^2\frac{1}{n}\sum_{i=1}^n \langle X_i \barX_i\rangle_{k,t;\eps}\bigg], \label{eq126}\\
\calY_{2,k}&:=\sqrt{\gamma_k(t)}\bbE\bigg[\frac{1}{m}\sum_{\mu=1}^m W_{\mu}^{(k)} \bigg\langle [\bA \bar{\bX}]_{\mu}\frac{1}{n}\sum_{i=1}^n X_i \barX_i \bigg\rangle_{k,t;\eps}\bigg] \label{eq127}.
\end{align}
By the law of large numbers, $\frac{1}{m}\sum_{\mu=1}^m \big(W_{\mu}^{(k)}\big)^2=1+o_n(1)$ almost surely, so we have
\begin{align}
\calY_{1,k}&:=\frac{1}{n}\bbE\bigg[\sum_{i=1}^n \langle X_i \barX_i\rangle_{k,t;\eps}\bigg]+o_n(1) \label{eq128b}\\
&=\frac{1}{n}\sum_{i=1}^n \bbE\big[\big(\langle \barX_i(S_i+\barX_i) \rangle_{k,t;\eps}\big)\big]+o_n(1) \label{eq128a}\\
&=\frac{1}{n}\sum_{i=1}^n \bbE\big[\langle \barX_i^2 \rangle_{k,t;\eps}\big]+ \frac{1}{n}\sum_{i=1}^n  \bbE[\langle \barX_i\rangle_{k,t;\eps} S_i]+o_n(1)\\
&=\frac{1}{n}\sum_{i=1}^n \bbE\big[\langle \barX_i^2 \rangle_{k,t;\eps}\big]+ \bbE\big[\langle \barq_{\bX,\bS}\rangle_{k,t;\eps}\big]+o_n(1),
\end{align}
where
\begin{align}
\barq_{\bX,\bS}:=\frac{1}{n}\sum_{i=1}^n S_i\barX_i.
\end{align} Here, \eqref{eq128a} follows from $X_i=\barX_i+S_i$.

Now, for any $a, b$, by Cauchy-Schwarz inequality observe that
\begin{align}
\bbE[\langle ab \rangle_{k,t;\eps}]&=\bbE[\langle a \rangle_{k,t;\eps}]\bbE[\langle b\rangle_{k,t;\eps}]+ \bbE[\langle (a-\bbE[\langle a \rangle_{k,t;\eps}])b \rangle_{k,t;\eps}]\\
&= \bbE[\langle a \rangle_{k,t;\eps}]\bbE[\langle b\rangle_{k,t;\eps}]+ O\bigg(\bbE\bigg[\sqrt{\langle (a-\bbE[\langle a \rangle_{k,t;\eps}]\big)^2 \rangle_{k,t;\eps}\langle b^2 \rangle_{k,t;\eps}}\bigg]\bigg)\\
&= \bbE[\langle a \rangle_{k,t;\eps}]\bbE[\langle b\rangle_{k,t;\eps}]+ O\bigg(\sqrt{\bbE\big[\langle (a-\bbE[\langle a \rangle_{k,t;\eps}]\big)^2 \rangle_{k,t;\eps}\big]\bbE\big[\langle b^2 \rangle_{k,t;\eps}\big]}\bigg) \label{barbier}.
\end{align}
Let $a=\frac{1}{n}\sum_{i=1}^n X_i \barX_i$ and $b=W_\mu^{(k)}[\bA \bar{\bX}]_{\mu}$, then we have
\begin{align}
\bbE\bigg[W_{\mu}^{(k)}]\bigg\langle [\bA \bar{\bX}]_{\mu}\frac{1}{n}\sum_{i=1}^n X_i \barX_i\bigg\rangle_{k,t;\eps}\bigg]=\bbE[\langle ab \rangle_{k,t;\eps}] \label{toiyeu}.
\end{align}
On the other hand,  by Cauchy-Schwarz inequality, we also have
\begingroup
\allowdisplaybreaks
\begin{align}
\bbE[\langle b^2 \rangle_{k,t;\eps}]&=\bbE\bigg[\bigg\langle\bigg( W_\mu^{(k)}[\bA \bar{\bX}]_{\mu}\bigg)^2\bigg \rangle_{k,t;\eps}\bigg]\\
&\leq \bbE\bigg[\sqrt{\big\langle \big(W_\mu^{(k)}\big)^4 \big \rangle_{k,t;\eps}\big\langle \big([\bA \bar{\bX}]_{\mu}\big)^4\big\rangle_{k,t;\eps} }\bigg]\\
&\leq \sqrt{\bbE\big[\big\langle \big(W_\mu^{(k)}\big)^4 \big \rangle_{k,t;\eps}\big\langle \big( [\bA \bar{\bX}]_{\mu}\big)^4\big\rangle_{k,t;\eps}\big] }\\
&\leq \sqrt[4]{\bbE\big[\big\langle \big(W_\mu^{(k)}\big)^4 \big \rangle_{k,t;\eps}^2\big] \bbE\big[\big\langle \big( [\bA \bar{\bX}]_{\mu}\big)^4\big\rangle_{k,t;\eps}^2\big] }\\
&\leq \sqrt[4]{\bbE\big[ \big(W_\mu^{(k)}\big)^8 \big] \bbE\big[\big\langle \big([\bA \bar{\bX}]_{\mu}\big)^8\big\rangle_{k,t;\eps}\big] }\\
&=\sqrt[4]{ 105\bbE\big[\big\langle \big( [\bA \bar{\bX}]_{\mu}\big)^8\big\rangle_{k,t;\eps}\big] }\\
&=\sqrt[4]{ 105\bbE\big[\big\langle \big( [\bA \bX]_{\mu} -[\bA \bS]_{\mu}\big)^8\big\rangle_{k,t;\eps}\big] }\\
&=\sqrt[4]{ 105\bbE\bigg[\bigg\langle \sum_{i=0}^8 {8 \choose i}(-1)^i [\bA \bX]_{\mu}^i  [\bA \bS]_{\mu}^{8-i}\bigg\rangle_{k,t;\eps}\bigg] }\\
&=	\sqrt[4]{ 105 \sum_{i=0}^8 {8 \choose i}(-1)^i \bbE\big[\big\langle[\bA \bX]_{\mu}^i\big\rangle_{k,t;\eps}  [\bA \bS]_{\mu}^{8-i}\big] }\\
&\leq \sqrt[4]{ 105 \sum_{i=0}^8 {8 \choose i}\sqrt{ \bbE\big[\big\langle[\bA \bX]_{\mu}^{2i} \big\rangle_{k,t;\eps}\big]  \bbE\big[[\bA \bS]_{\mu}^{2(8-i)}\big]} }\\
&=\sqrt[4]{ 105 \sum_{i=0}^8 {8 \choose i} \sqrt{\bbE\big[[\bA \bS]_{\mu}^{2i}\big]  \bbE\big[[\bA \bS]_{\mu}^{2(8-i)}\big] }} \\
&=\sqrt[4]{ \frac{105}{n^8} \sum_{i=0}^8 {8 \choose i} \sqrt{\bbE\big[[\sqrt{n}\bA \bS]_{\mu}^{2i}\big]  \bbE\big[[\sqrt{n}\bA \bS]_{\mu}^{2(8-i)}\big] }} 
\label{layloi}.
\end{align}
\endgroup
Now, we have
\begin{align}
[\sqrt{n}\bA \bS]_{\mu}=\sum_{i=1}^n \sqrt{n} A_{\mu,i}S_i.
\end{align}
Note that
\begin{align}
\var(\sqrt{n} A_{\mu,i} S_i)&=n\bbE[A_{\mu,i}^2]\bbE_{S_i \sim \tilP_0}[S_i^2]\\
&=\frac{n}{m}\bbE_{S_i \sim \tilP_0}[S_i^2]\\
&=\frac{n}{\delta n^{\alpha}} \frac{n^{\alpha}}{n}\bbE_{S\sim P_0}[S^2]\\
&=\frac{1}{\delta}\bbE_{S\sim P_0} [S^2].
\end{align}
Hence, by the central limit theorem, we have
\begin{align}
[\sqrt{n}\bA \bS]_{\mu} \to \calN\bigg(0, \frac{1}{\delta}\bbE_{S \sim P_0}[S^2]\bigg).
\end{align}
It follows that $\bbE\big[[\sqrt{n}\bA \bS]_{\mu}^{2(8-i)}\big]$ and $ \bbE\big[[\sqrt{n}\bA \bS]_{\mu}^{2i}\big] $ are bounded for each $i \in [8]$. Hence, $\bbE[\langle b^2 \rangle_{k,t;\eps}]$ goes to zero uniformly in $\mu$ as $n \to \infty$ by observing \eqref{layloi}.  

Furthermore, we have
\begingroup
\allowdisplaybreaks
\begin{align}
\bbE\big[\langle (a-\bbE[\langle a \rangle_{k,t;\eps}]\big)^2\rangle_{k,t;\eps} \big]&=\bbE[\langle a^2 \rangle_{k,t;\eps}]-\bbE\big[\big(\langle a \rangle_{k,t;\eps}\big)^2\big]\\
& \leq \bbE[\langle a^2 \rangle_{k,t;\eps}]\\
&= \bbE\bigg[\bigg\langle \bigg(\frac{1}{n}\sum_{i=1}^n X_i \barX_i\bigg)^2\bigg \rangle_{k,t;\eps}\bigg]\\
&\leq  \bbE\bigg[\bigg\langle \bigg(\frac{1}{n}\sum_{i=1}^n |X_i| |\barX_i|\bigg)^2\bigg \rangle_{k,t;\eps}\bigg]\\
&\leq  \bbE\bigg[\bigg\langle \bigg(\frac{1}{n}\sum_{i=1}^n |X_i| 2s_{\max}\bigg)^2\bigg \rangle_{k,t;\eps}\bigg] \label{eqsmax}\\
&=4 s_{\max}^2\bbE\bigg[\bigg(\frac{1}{n}\sum_{i=1}^n |X_i| \bigg)^2\bigg]\\
&\leq \frac{4s_{\max}^2}{n}\bbE\bigg[\sum_{i=1}^n X_i^2\bigg]\\
&= 4s_{\max}^2 \bbE_{S \sim \tilP_0} [S^2]\\
&=4s_{\max}^2 \frac{n^{\alpha}}{n}\bbE_{S \sim P_0} [S^2]\\
&=O_n\bigg(\frac{1}{n^{1-\alpha}}\bigg) \to 0 \label{eq153}
\end{align} 
\endgroup
as $0\leq \alpha<1$, where \eqref{eqsmax} follows from the fact that $|\barX_i|=|X_i-S_i| \leq |X_i|+|S_i|  \leq 2s_{\max}$.

From \eqref{barbier}, \eqref{toiyeu}, \eqref{layloi}, and \eqref{eq153}, we obtain
\begin{align}
\bbE\bigg[W_{\mu}^{(k)}\bigg\langle [\bA \bar{\bX}]_{\mu}\frac{1}{n}\sum_{i=1}^n X_i \barX_i\bigg\rangle_{k,t;\eps}\bigg]&= \bbE[\langle a \rangle_{k,t;\eps}]\bbE[\langle b \rangle_{k,t;\eps}]+o_n(1)\\
&=\bbE\bigg[\bigg \langle \frac{1}{n}\sum_{i=1}^n X_i \barX_i \bigg \rangle_{k,t;\eps} \bigg]\bbE\bigg[W_{\mu}^{(k)}\langle [\bA\bar{\bX}]_{\mu}\rangle_{k,t;\eps}\bigg]+o_n(1),
\end{align} where $o_n(1) \to 0$ uniformly in $\mu$.

It follows that
\begin{align}
\calY_{2,k}&=\sqrt{\gamma_{k}(t)}\frac{1}{m}\sum_{\mu=1}^m\bbE\bigg[W_{\mu}^{(k)}\bigg\langle [\bA \bar{\bX}]_{\mu}\frac{1}{n}\sum_{i=1}^n X_i \barX_i\bigg\rangle_{k,t;\eps}\bigg] \\
&=\sqrt{\gamma_{k}(t)} \bbE\bigg[\bigg \langle \frac{1}{n}\sum_{i=1}^n X_i \barX_i \bigg \rangle_{k,t;\eps}\bigg] \bigg(\frac{1}{m}\sum_{\mu=1}^m  \bbE\bigg[W_{\mu}^{(k)}\langle [\bA\bar{\bX}]_{\mu}\rangle_{k,t;\eps}\bigg]\bigg)+o_n(1) \label{final1}.
\end{align}
Now, by \citep[Eq.~(26)]{BarbierALT2016}, we have
\begin{align}
{\rm{ymmse}_{k,t;\eps}}=\frac{1}{m \sqrt{\gamma_k(t)}}\sum_{\mu=1}^m \bbE\bigg[W_{\mu}^{(k)}\langle [\bA\bar{\bX}]_{\mu}\rangle_{k,t;\eps}\bigg]. \label{final2}
\end{align}
Hence, from \eqref{final1} and \eqref{final2}, we obtain
\begin{align}
\calY_{2,k}&=\gamma_k(t)\bbE\bigg[\bigg \langle \frac{1}{n}\sum_{i=1}^n X_i \barX_i \bigg \rangle_{k,t;\eps}\bigg]{\rm{ymmse}_{k,t;\eps}}+o_n(1)\\
&=\gamma_k(t){\rm{ymmse}_{k,t;\eps}} \mathcal{\tilY}_{1,k}+o_n(1) \label{eq128}, 
\end{align} where
\begin{align}
\mathcal{\tilY}_{1,k}=\bbE\bigg[\bigg \langle \frac{1}{n}\sum_{i=1}^n X_i \barX_i \bigg \rangle_{k,t;\eps}\bigg].
\end{align}

It follows from \eqref{eq125}--\eqref{eq128} and \eqref{eq128} that
\begin{align}
{\rm{ymmse}_{k,t;\eps}}&=\calY_{1,k}-\calY_{2,k}\\
&=\mathcal{\tilY}_{1,k}+o_n(1)-\mathcal{\tilY}_{1,k}\gamma_k(t) {\rm{ymmse}_{k,t;\eps}}+o_n(1)\\
&=\mathcal{\tilY}_{1,k}-\mathcal{\tilY}_{1,k}\gamma_k(t) {\rm{ymmse}_{k,t;\eps}}+o_n(1).
\end{align}
This leads to
\begin{align}
{\rm{ymmse}_{k,t;\eps}}=\frac{\mathcal{\tilY}_{1,k}}{1+\gamma_k(t) \mathcal{\tilY}_{1,k}}+o_n(1).
\end{align}
Then, it holds that
\begin{align}
{\rm{ymmse}_{k,t;\eps}}-\frac{{\rm{mmse}_{k,t;\eps}}n^{\alpha-1}}{1+\gamma_k(t){\rm{mmse}_{k,t;\eps}}n^{\alpha-1} }&=\frac{\mathcal{\tilY}_{1,k}}{1+\gamma_k(t) \mathcal{\tilY}_{1,k}}-\frac{{\rm{mmse}_{k,t;\eps}}n^{\alpha-1}}{1+\gamma_k(t){\rm{mmse}_{k,t;\eps}}n^{\alpha-1} } +o_n(1)\label{eqkey}.
\end{align}
Now, observe that
\begin{align}
\big|\mathcal{\tilY}_{1,k}-{\rm{mmse}_{k,t;\eps}}n^{\alpha-1}\big|&=\bigg|\frac{1}{n}\sum_{i=1}^n \bbE\big[\langle \barX_i^2 \rangle_{k,t;\eps}\big]-\frac{1}{n}\sum_{i=1}^n \bbE\big[\langle \barX_i \rangle_{k,t;\eps}^2\big]+\bbE[\langle\barq_{\bX,\bS}\rangle_{k,t;\eps}]\bigg|+o_n(1)\\
&\leq \bigg|\frac{1}{n}\sum_{i=1}^n \bbE\big[\langle \barX_i^2 \rangle_{k,t;\eps}\big]-\frac{1}{n}\sum_{i=1}^n \bbE\big[\langle \barX_i \rangle_{k,t;\eps}^2\big]\bigg|+\bigg|\bbE[\langle\barq_{\bX,\bS}\rangle_{k,t;\eps}]\bigg|+o_n(1)\\
& \leq \frac{1}{n}\sum_{i=1}^n \bbE\big[\langle \barX_i^2 \rangle_{k,t;\eps}\big]+\frac{1}{n}\sum_{i=1}^n \bbE[|S_i \langle \barX_i \rangle_{k,t;\eps}|]+o_n(1)\\
&=\sum_{i=1}^n \bbE\big[\langle (X_i-S_i)^2 \rangle_{k,t;\eps}\big]+\frac{1}{n}\sum_{i=1}^n \bbE[|S_i \langle (X_i-S_i) \rangle_{k,t;\eps}|]+o_n(1)\\
& \leq \frac{2}{n}\sum_{i=1}^n \bbE\big[\langle X_i^2 + S_i^2 \rangle_{k,t;\eps}\big]+\frac{1}{n}\sum_{i=1}^n\sqrt{ \bbE[S_i^2] \bbE[\langle (X_i-S_i) \rangle_{k,t;\eps}^2]} \label{cauchy}\\
& \leq \frac{2}{n}\sum_{i=1}^n \bbE\big[\langle X_i^2 + S_i^2 \rangle_{k,t;\eps}\big]+\frac{1}{n}\sum_{i=1}^n\sqrt{ \bbE[S_i^2] \bbE[\langle 2(S_i^2+X_i^2) \rangle_{k,t;\eps}]} \\
&=6 \bbE_{S \sim \tilP_0} \bbE[S^2] \label{iid}\\
&= \frac{6 n^{\alpha}}{n}\bbE_{S \sim P_0} \bbE[S^2]\\
&:=f(n) \label{bound},
\end{align} where $f(n)=O\big(\frac{1}{n^{1-\alpha}}\big)=o(1) $ uniformly in $k,t$ if $0 \leq \alpha <1$. Here, \eqref{cauchy} follows from Cauchy–Schwarz inequality, \eqref{iid} follows from the i.i.d. assumption of the sequence $\{S_i\}_{i=1}^n$ and $\{X_i\}_{i=1}^n$ under $\tilP_0$, and \eqref{bound} follows from the assumption that $\bbE_{S \sim P_0}[S^4]<\infty$.

Now, let
\begin{align}
g_{k,t}(x):=\frac{x}{1+\gamma_k(t)x}.
\end{align} It is easy to see that $g_{k,t}(x)$ is an increasing function for $x \geq 0$. More over, we have
\begin{align}
0<g_{k,t}'(x)=\frac{1}{(1+\gamma_k(t)x)^2} \leq 1
\end{align} uniformly in $k,t$ for all $x \geq 0$ (since $\gamma_k(t)\geq 0$ uniformly in $k,t$). Hence, we have
\begin{align}
&\bigg|{\rm{ymmse}_{k,t;\eps}}-\frac{{\rm{mmse}_{k,t;\eps}}n^{\alpha-1}}{1+\gamma_k(t){\rm{mmse}_{k,t;\eps}}n^{\alpha-1} }\bigg|\nn\\
&\qquad=\bigg|\frac{\mathcal{\tilY}_{1,k}}{1+\gamma_k(t) \mathcal{\tilY}_{1,k}}-\frac{{\rm{mmse}_{k,t;\eps}}n^{\alpha-1}}{1+\gamma_k(t){\rm{mmse}_{k,t;\eps}}n^{\alpha-1} }\bigg|+o_n(1)\\
&\qquad\leq \bigg|g_{k,t}({\rm{mmse}_{k,t;\eps}}n^{\alpha-1}\pm f(n))-g_{k,t}({\rm{mmse}_{k,t;\eps}}n^{\alpha-1})\bigg|+o_n(1)\\
&\qquad=|g_{k,t}'(\theta) f(n)|+o_n(1)
\end{align} for some $\theta >0$ by Taylor's expansion. This means that
\begin{align}
\bigg|{\rm{ymmse}_{k,t;\eps}}-\frac{{\rm{mmse}_{k,t;\eps}}n^{\alpha-1}}{1+\gamma_k(t){\rm{mmse}_{k,t;\eps}}n^{\alpha-1} }\bigg| \leq \tilf(n) \label{cubest}
\end{align} uniformly in $k,t$ where $\tilf(n):=f(n)+o_n(1)$.

Then, it holds that
\begin{align}
&\bigg|\int_{a_n}^{b_n}d\eps\frac{1}{K_n}\sum_{k=1}^{K_n}\int_0^1 dt \frac{d\gamma_k(t)}{dt}\bigg({\rm{ymmse}_{k,t;\eps}} 
-\frac{\rm{mmse}_{k,t;\eps}}{n^{1-\alpha}+\gamma_k(t)\rm{mmse}_{k,t;\eps}}  \bigg)\bigg|\nn\\
& \leq \int_{a_n}^{b_n}d\eps \frac{1}{K_n}\sum_{k=1}^{K_n}\bigg|\int_0^1 dt \frac{d\gamma_k(t)}{dt}\bigg| \tilf(n)\\
&=(b_n-a_n) \tilf(n) \frac{1}{K_n}\sum_{k=1}^{K_n} \big|\gamma_k(1)-\gamma_k(0)\big|\\
&=(b_n-a_n) \tilf(n) \frac{1}{K_n}\sum_{k=1}^{K_n}\frac{1}{\Delta_n}\\
&=(b_n-a_n) \tilf(n) \frac{1}{\Delta_n}\\
&=o\bigg(\frac{b_n-a_n}{\Delta_n}\bigg)
\end{align} as $a_n, b_n \to 0$. Hence, we have
\begin{align}
\int_{a_n}^{b_n}\frac{1}{K_n}\sum_{k=1}^{K_n}\int_0^1 dt \frac{d\gamma_k(t)}{dt}\bigg({\rm{ymmse}_{k,t;\eps}} 
-\frac{\rm{mmse}_{k,t;\eps}}{n^{1-\alpha}+\gamma_k(t)\rm{mmse}_{k,t;\eps}}  \bigg)=o\bigg(\frac{b_n-a_n}{\Delta_n}\bigg) \label{bunhi2}.
\end{align}
From \eqref{bunhi1a} and \eqref{bunhi2}, we obtain \eqref{keyconcen} for all $0\leq \alpha\leq 1$.
\section*{Appendix B.} \label{basiscexnew:proof} 
In this Appendix, we provide a proof for Lemma \ref{basiscexnew}. \\

Observe that
	\begin{align}
	{\rm{mmse}}_{k,t;\eps}-{\rm{mmse}}_{k,0;\eps}=\int_0^t \frac{d\rm{mmse}_{k,\nu;\eps}}{d\nu} d\nu \label{ufac0}.
	\end{align}
	Now, we have
	\begin{align}
	n^{\alpha}\frac{d\rm{mmse}_{k,\nu;\eps}}{d\nu}&=\frac{d}{d\nu}\bbE[\|\langle \bX\rangle_{k,\nu;\eps}-\bS\|^2]\\
	&=\frac{d}{d\nu}\bbE[\|\langle \bX\rangle_{k,\nu;\eps}\|^2]-2 \bbE\bigg[\bS^T \frac{d}{d\nu}\langle \bX\rangle_{k,\nu;\eps}\bigg]+ \frac{d}{d\nu}\bbE[\|\bS\|^2]\\
	&=2\frac{d}{d\nu}\bbE[\|\bS\|^2]-2 \bbE\bigg[\bS^T \frac{d}{d\nu}\langle \bX\rangle_{k,\nu;\eps}\bigg]\\
	&=-2 \bbE\bigg[\bS^T \frac{d}{d\nu}\langle \bX\rangle_{k,\nu;\eps}\bigg] \label{C1}. 
	\end{align}
	Now, it is easy to see that
	\begin{align}
	\frac{d}{d\nu}\langle \bX\rangle_{k,\nu;\eps}&=\langle \bX \rangle_{k,\nu;\eps}\bigg\langle \frac{d \calH_{k,\nu;\eps}(\bX,\bTheta)}{d\nu}\bigg\rangle-\bigg\langle \bX \frac{d \calH_{k,\nu;\eps}(\bX,\bTheta)}{d\nu}\bigg \rangle_{k,\nu;\eps}.
	\end{align}
	Define
	\begin{align}
	q_{\bx,\bs}:=\frac{1}{n^{\alpha}}\sum_{i=1}^n x_i s_i,
	\end{align} which is a normalized overlap between $\bx$ and $\bs$.
	
	Let $\bX'$ is a replica of $\bX$, i.e. $P_{k,\nu;\eps}(\bx,\bx'|\btheta)=P_{k,\nu;\eps}(\bx|\btheta)P_{k,\nu;\eps}(\bx'|\btheta)$, then it holds that
	\begin{align}
	\bbE\bigg[\bS^T \frac{d}{d\nu}\langle \bX\rangle_{k,\nu;\eps}\bigg]&=n^{\alpha}\bbE\bigg[\langle q_{\bX,\bS} \rangle_{k,\nu;\eps}\bigg\langle \frac{d \calH_{k,\nu;\eps}(\bX,\bTheta)}{d\nu}\bigg\rangle  -\bigg\langle q_{\bX,\bS} \frac{d \calH_{k,\nu;\eps}(\bX,\bTheta)}{d\nu}\bigg \rangle_{k,\nu;\eps}\bigg]\\
	&=n^{\alpha}\bbE\bigg[\bigg\langle q_{\bX,\bS}\bigg( \frac{d \calH_{k,\nu;\eps}(\bX',\bTheta)}{d\nu} -\frac{d \calH_{k,\nu;\eps}(\bX,\bTheta)}{d\nu}\bigg)\bigg \rangle_{k,\nu;\eps}\bigg].
	\end{align}
	Now, observe that
	\begin{align}
	\frac{d \calH_{k,\nu;\eps}(\bX,\bTheta)}{d\nu}=\frac{d}{d\nu} h\bigg(\bX,\bS,\bA,\bW^{(k)}, \frac{K_n}{\gamma_k(\nu)}\bigg)+\frac{d}{d\nu} h_{\rm{mf}}\bigg(\bX,\bS,\tilbW^{(k)},\frac{K_n}{\lambda_k(\nu)}\bigg),
	\end{align} where
	\begin{align}
	\frac{d}{d\nu} h\bigg(\bX,\bS,\bA,\bW^{(k)}, \frac{K_n}{\gamma_k(\nu)}\bigg)=\frac{1}{2K_n} \bigg(\frac{d\gamma_k(\nu)}{d\nu}\bigg)\bigg(\sum_{\mu=1}^m [\bA \overline{\bX}]_{\mu}^2-\sqrt{\frac{K_n}{\gamma_k(\nu)}}\sum_{\mu=1}^m[\bA \overline{\bX}]_{\mu} W_{\mu}^{(k)}\bigg),
	\end{align}
	and
	\begin{align}
	\frac{d}{d\nu} h_{\rm{mf}}\bigg(\bX,\bS,\tilbW^{(k)},\frac{K_n}{\lambda_k(\nu)}\bigg)=\frac{1}{2K_n} \bigg(\frac{d\lambda_k(\nu)}{d\nu}\bigg)\bigg(\sum_{i=1}^n\barX_i^2-\sqrt{\frac{K_n}{\lambda_k(\nu)}}\sum_{i=1}^n \barX_i \tilW_i^{(k)}\bigg).
	\end{align}
	Using the fact that $\bbE[W_\mu^{(k)}]\sim \calN(0,1)$ and $\bbE[\tilW_\mu^{(k)}]\sim \calN(0,1)$ and the fact that $\bbE[Zf(Z)]=\bbE[f'(Z)]$ for $Z \sim \calN(0,1)$, we finally have
	\begin{align}
	&\bbE\bigg[\bS^T \frac{d}{d\nu}\langle \bX\rangle_{k,\nu;\eps}\bigg]=\frac{n^{\alpha}}{2K_n}\bigg(\frac{d\gamma_k(\nu)}{d\nu}\bigg)\bbE\big[g(\bX',\bS)-g(\bX,\bS)\big] \label{ufact},
	\end{align}
	where
	\begin{align}
	g(\bx,\bs):=\sum_{\mu=1}^m \langle [\bA \overline{\bx}]_{\mu} q_{\bx,\bs}\rangle  \langle [\bA \overline{\bx}]_{\mu} \rangle -\frac{\delta n^{\alpha-1}}{(1+\gamma_k(\nu) E_k)^2}\sum_{i=1}^n \langle \barx_i q_{\bx,\bs}\rangle_{k,\nu;\eps} \langle \barx_i \rangle_{k,\nu;\eps}.
	\end{align}
	Hence, we have
	\begin{align}
	\big|\bbE\big[g(\bX,\bS)\big]\big|&\leq \bbE\bigg[\bigg|\sum_{\mu=1}^m \langle [\bA \overline{\bX}]_{\mu} q_{\bX,\bS}\rangle_{k,\nu;\eps}  \langle [\bA \overline{\bX}]_{\mu} \rangle_{k,\nu;\eps}\bigg|\bigg]\nn\\ 
	&\qquad + \frac{\delta n^{\alpha-1}}{(1+\gamma_k(\nu) E_k)^2}\bbE\bigg[\bigg|\sum_{i=1}^n \langle \barX_i q_{\bX,\bS}\rangle_{k,\nu;\eps}\langle \barX_i \rangle_{k,\nu;\eps}\bigg|\bigg]\\
	& \leq \sum_{\mu=1}^m \bbE\bigg[\big|\langle [\bA \overline{\bX}]_{\mu} q_{\bX,\bS}\rangle_{k,\nu;\eps}  \langle [\bA \overline{\bX}]_{\mu} \rangle_{k,\nu;\eps}\big|\bigg]\nn\\ 
	&\qquad + \delta n^{\alpha-1}\sum_{i=1}^n \bbE\bigg[\big|\langle \barX_i q_{\bX,\bS}\rangle_{k,\nu;\eps}\langle \barX_i \rangle_{k,\nu;\eps}\big|\bigg]\label{C2}. 
	\end{align}
	Now, by using Cauchy's and Cauchy-Schwarz's inequalities, we have
	\begin{align}
	&\bbE\bigg[\big|\langle [\bA \overline{\bX}]_{\mu} q_{\bX,\bS}\rangle_{k,\nu;\eps}  \langle [\bA \overline{\bX}]_{\mu} \rangle_{k,\nu;\eps}\big|\bigg]\nn\\
	&\qquad \leq \frac{1}{2}\bbE\big[ \langle [\bA \overline{\bX}]_{\mu} q_{\bX,\bS}\rangle_{k,\nu;\eps}^2\big]+ \frac{1}{2}\bbE\big[ \langle [\bA \overline{\bX}]_{\mu} \rangle_{k,\nu;\eps}^2\big]\\
	&\qquad \leq \frac{1}{2}\sqrt{\bbE[\langle [\bA \overline{\bX}]_{\mu} \rangle_{k,\nu;\eps}^4] \bbE[\langle q_{\bX,\bS}^4\rangle_{k,\nu;\eps}]}+\frac{1}{2}\sqrt{\bbE[\langle [\bA \overline{\bX}]_{\mu}^4 \rangle_{k,\nu;\eps}]} \label{D1}. 
	\end{align}
	On the other hand, we have
	\begin{align}
	\bbE[\langle [\bA \overline{\bX}]_{\mu} \rangle_{k,\nu;\eps}^4]&=\bbE[\langle [\bA (\bX-\bS)]_{\mu} \rangle_{k,\nu;\eps}^4]\\
	&\leq 8 \bigg(\bbE\big[\langle [\bA\bX]_{\mu} \rangle_{k,\nu;\eps}^4\big] + \bbE \big[ [\bA\bS]_{\mu}^4\big]\bigg) \label{D2}\\
	&= 16 \bbE \big[ [\bA\bS]_{\mu}^4\big]\\
	&= 16 \bigg(\sum_{i=1}^n \bbE[A_{\mu, i}^4]\bbE[S_i^4]+ 6 n(n-1) \sum_{i=1}^n \bbE[A_{\mu,i}^2]\bbE[S_i^2]\bigg)\\
	&= 16 \bigg( \frac{n}{m^2} \frac{k}{n}\bbE_{S \sim P_0}[S^4]+ 6 n(n-1) \frac{n}{m}  \frac{k}{n} \bbE_{S \sim P_0} [S^2]\bigg)\\
	&=O(n^2) \label{D4a},
	\end{align} where \eqref{D2} follows from  $(a+b)^4\leq 8 (a^4+b^4)$.
	
	On the other hand, since
	\begin{align}
	\big|q_{\bx,\bs}\big|&=\frac{1}{n^{\alpha}}\bigg|\sum_{i=1}^n x_i s_i\bigg|\\
	&\leq \frac{1}{n^{\alpha}}\sum_{i=1}^n\big|x_i s_i\big|\\
	&\leq \frac{1}{n^{\alpha}} n s_{\max}^2\\
	&=  s_{\max}^2 n^{1-\alpha} \label{D4}.
	\end{align}
From \eqref{D1}, \eqref{D4a}, and \eqref{D4}, we obtain
\begin{align}
\bbE\bigg[\big|\langle [\bA \overline{\bX}]_{\mu} q_{\bX,\bS}\rangle_{k,\nu;\eps}  \langle [\bA \overline{\bX}]_{\mu} \rangle_{k,\nu;\eps}\big|\bigg]=O\big(n^{(3-\alpha)/2}\big) \label{Bc}, 
\end{align} where the constant does not depend on $\nu$.

Similarly, we have
\begin{align}
\bbE\bigg[\big|\langle \barX_i q_{\bX,\bS}\rangle_{k,\nu;\eps}\langle \barX_i \rangle_{k,\nu;\eps}\big|\bigg]=O(1) \label{Bc2},
\end{align}	 where the constant does not depend on $i$.

From \eqref{C2}, \eqref{Bc}, and \eqref{Bc2}, we obtain
\begin{align}
\big|\bbE\big[g(\bX,\bS)\big]\big| =O(n^{(3+\alpha)/2})  \label{Bc3}. 
\end{align}
From \eqref{ufact} and \eqref{Bc3}, we obtain
\begin{align}
\bigg|\bbE\bigg[\bS^T \frac{d}{d\nu}\langle \bX\rangle_{k,\nu;\eps}\bigg]\bigg|\leq O\bigg(\frac{n^{\alpha}}{K_n} \bigg|\frac{d\gamma_k(\nu)}{d\nu}\bigg|n^{(3+\alpha)/2}\bigg) \label{ufact2},
\end{align} where the constant does not depend on $\nu$.

From \eqref{ufac0}, \eqref{C1}, and \eqref{ufact2}, for some constant $C$, we have
\begin{align}
\bigg|{\rm{mmse}}_{k,t;\eps}-{\rm{mmse}}_{k,0;\eps}\bigg|&\leq \int_0^t \bigg|\frac{d\rm{mmse}_{k,\nu;\eps}}{d\nu}\bigg| d\nu \label{ufac0b}\\
&\leq \int_0^1 \bigg|\frac{d\rm{mmse}_{k,\nu;\eps}}{d\nu}\bigg| d\nu \label{ufac0c}\\
&\leq C \int_0^1 \frac{n^{\alpha}}{K_n} \bigg|\frac{d\gamma_k(\nu)}{d\nu}\bigg|n^{(3+\alpha)/2} d\nu\\
&= -C \int_0^1 \frac{n^{\alpha}}{K_n} \bigg(\frac{d\gamma_k(\nu)}{d\nu}\bigg)n^{(3+\alpha)/2} d\nu\\
&=O\bigg(\frac{n^{(3/2)(1+\alpha)}}{K_n \Delta_n}\bigg).
\end{align}
This concludes our proof of Lemma \ref{basiscexnew}.

\section*{Appendix C.}\label{GSLLN:proof}
In this Appendix, we provide a proof for Lemma \ref{GSLLN}.\\

First, we recall the following result.
\begin{lemma}[Abel-Dini Theorem] \label{ADthm} If $\sum_{n=1}^{\infty} a_n$ converges, then with $T_n$ as the $n$-th tail, $\sum_{n=1}^{\infty} \frac{a_n}{T_n^{1+\alpha}}$ converges if and only if $\alpha<1/2$.
\end{lemma}

The proof is based on the proofs of \citep[Theorem 1.1]{Fazekas2001AGA} and \citep[Theorem 2.1]{Fazekas2001AGA}. To begin with, we generalize H\'{a}jek-R\'{e}nyi	type maximal inequality \citep[Theorem 1.1]{Fazekas2001AGA} that under the condition \eqref{cond1}, for any non-decreasing sequence of positive numbers $\{\beta_n\}_{n=1}^{\infty}$, it holds that
\begin{align}
\bbE\bigg[\max_{1\leq l\leq n}\bigg|\frac{S_l}{\beta_l d_l}\bigg|^{r} \bigg]\leq 4 \sum_{l=1}^n \frac{\nu_l}{\beta_l^r d_l^{r-1/(1+\rho)}} \label{G1}.
\end{align}
To show \eqref{G1}, we can suppose that $\beta_1=1$. Let $c=2^{1/r}$. Consider the sets
\begin{align}
A_i=\bigg\{k:c^i\leq \beta_k d_k <c^{i+1} \bigg\}, \qquad \forall i=0,1,2,\cdots \label{defAi},
\end{align}
Denote by $i(n)$ the index of the last non-empty $A_j$ such that $A_j \subset [n]$. It is clear that $i(n)\geq 0$ since $A_0 \neq \emptyset$ by $\beta_1=\nu_1=1$. Let $k(i)=\max\{k:k \in A_i\}, i=0,1,2,\cdots $ if $A_i$ is non-empty, while $k(i)=k(i-1)$ if $A_i$ is empty, and let $k(-1)=0$. Let
\begin{align}
\delta_l:=d_{k(l)}^{1/(1+\rho)}\sum_{j=k(l-1)+1}^{k(l)}\nu_j, \qquad l=0,1,2,\cdots,
\end{align} where $\delta_l$ is considered to be zero if $A_l$ is empty. We have
\begingroup
\allowdisplaybreaks
\begin{align}
\bbE\bigg[\max_{1\leq l \leq n}\bigg|\frac{S_l}{\beta_l d_l} \bigg|^r\bigg]&\leq \sum_{i=0}^{i(n)} \bbE\bigg[\max_{l\in A_i}\bigg|\frac{S_l}{\beta_l d_l} \bigg|^r\bigg]\\
&\leq \sum_{i=0}^{i(n)} c^{-ir}\bbE\bigg[\max_{l\in A_i}\big|S_l\big|^r\bigg]\\
&\leq \sum_{i=0}^{i(n)} c^{-ir}\bbE\bigg[\max_{k \leq k(i) }\big|S_k\big|^r\bigg]\\
&=\sum_{i=0}^{i(n)} c^{-ir} d_{k(i)}^{1/(1+\rho)} \sum_{k=1}^{k(i)} \nu_k\\
&=\sum_{i=0}^{i(n)} c^{-ir} \sum_{l=0}^i \delta_l\\
&=\sum_{l=0}^{i(n)} \delta_l \sum_{i=l}^{i(n)} c^{-ir}\\
&=\sum_{l=0}^{i(n)} \delta_l \sum_{i=l}^{\infty} c^{-ir}\\
&= \frac{1}{1-c^{-r}} \sum_{l=0}^{i(n)} \delta_l c^{-lr}\\
&=\frac{1}{1-c^{-r}} \sum_{l=0}^{i(n)}  c^{-lr}d_{k(l)}^{1/(1+\rho)}\sum_{j=k(l-1)+1}^{k(l)}\nu_j \\
&\leq \frac{c^r}{1-c^{-r}} \sum_{l=0}^{i(n)}   d_{k(l)}^{1/(1+\rho)}\sum_{j=k(l-1)+1}^{k(l)}\frac{\nu_j}{\big(\beta_j d_j\big)^r} \label{bato1}\\
&\leq \frac{c^r}{1-c^{-r}} \sum_{l=0}^{i(n)}   \sum_{j=k(l-1)+1}^{k(l)}\frac{\nu_j d_j^{1/(1+\rho)}}{\big(\beta_j d_j\big)^r} \label{bato2}\\
&=\frac{c^r}{1-c^{-r}} \sum_{l=0}^{i(n)}   \sum_{j=k(l-1)+1}^{k(l)}\frac{\nu_j}{\beta_j^r d_j^{r-1/(1+\rho)}} \\
&=\frac{c^r}{1-c^{-r}} \sum_{j=0}^n \frac{\nu_j}{\beta_j^r d_j^{r-1/(1+\rho)}}\\
&=4 \sum_{j=0}^n \frac{\nu_j}{\beta_j^r d_j^{r-1/(1+\rho)}} \label{gato},
\end{align}
\endgroup
 where \eqref{bato1} follows from $A_l=\{j: k(l-1)+1 \leq j \leq k(l)\}$ and $c^{-lr}\leq \frac{c^r}{(\beta_j d_j)^r}$ for all $j \in A_l$ which is achieved from \eqref{defAi}, \eqref{bato2} follows from $d_{k(l)}\leq d_j$ for all $j\in A_l$, and \eqref{gato} follows  from $c=2^{1/r}$.\\

Now, we return prove \eqref{G2}. As the proof of \citep[Theorem 2.1]{Fazekas2001AGA}, we can assume that $\nu_n>0$ for an infinite number of indices $n$. Otherwise, there exists an integer $n_0$ such that $\nu_n=0$ for all $n\geq n_0$, so from the condition \eqref{cond1}, we have $\bbE\big[\sup_{n\geq 1}|S_n|^r\big]<\infty$, so $\sup_{n\neq 1}|S_n|<\infty$ a.s., which easily gives \eqref{G2}. 	Now, set
\begin{align}
t_n&=\sum_{k=n}^{\infty}\frac{\nu_k}{b_k^r d_k^{r-1/(1+\rho)}}\\
\beta_n&=\max_{1\leq k\leq n} b_k t_k^{1/(2r)} \label{bag1}.
\end{align}
Then, it holds that
\begin{align}
\sum_{k=1}^{\infty} \frac{\nu_k}{d_k^{r-1/(1+\rho)}\beta_k^r}&\leq \sum_{k=1}^{\infty} \frac{\nu_k}{d_k^{r-1/(1+\rho)}b_k^r t_k^{1/2}} \label{bag2}\\
&< \infty \label{bag3}, 
\end{align} where \eqref{bag2} follows from \eqref{bag1}, and \eqref{bag3} follows from Lemma \ref{ADthm} for the convergence of series \citep{Fazekas2001AGA} with $\alpha=-1/2$. On the other hand, from \eqref{bag1}, $\beta_n$ is non-decreasing. In addition, from  \eqref{assum1} and \eqref{bag1}, we have $t_n \to 0$ as $n\to \infty$. Hence, for any $\eps>0$, there exists an $n_0(\eps)\in \bbZ^+$ such that 
\begin{align}
t_n^{1/(2r)}<\eps \label{bg1}
\end{align} for all $n>n_0(\eps)$.
It follows that
\begin{align}
\beta_n&=\max_{1\leq k\leq n} b_k t_k^{1/(2r)}\\
&\leq \max_{1\leq k\leq n_0(\eps)} b_k t_k^{1/(2r)} + \max_{n_0(\eps)< k\leq n} b_k t_k^{1/(2r)}\\
&\leq  \max_{1\leq k\leq n_0(\eps)} b_k t_k^{1/(2r)} + \eps \max_{n_0(\eps)< k\leq n} b_k \label{bg2}\\
&= \max_{1\leq k\leq n_0(\eps)} b_k t_k^{1/(2r)} + \eps b_n \label{bg3}\\
\end{align} where \eqref{bg2} follows from \eqref{bg1}, and \eqref{bg3} follows from the non-decreasing property of the sequence $\{b_n\}$. From \eqref{bg3}, we obtain
\begin{align}
0&\leq \frac{\beta_n}{b_n}\\
&\leq  \frac{1}{b_n}\max_{1\leq k\leq n_0(\eps)} b_k t_k^{1/(2r)}+ \eps\\
&\leq 2\eps \label{bg4},
\end{align} where \eqref{bg4} follows from $b_n \to \infty$. Hence, $\lim_{n\to \infty} \beta_n/b_n=0$. 

These facts imply that \eqref{G1} holds, which leads to 
\begin{align}
\bbE\bigg[\max_{l\geq 1}\bigg|\frac{S_l}{\beta_l d_l}\bigg|^{r} \bigg] &\leq 4 \sum_{l=1}^{\infty} \frac{\nu_l}{\beta_l^r d_l^{r-1/(1+\rho)}} \label{bag4}\\
&<\infty \label{bag5},
\end{align} where \eqref{bag4} follows from the monotone convergence theorem \citep{Billingsley}, and \eqref{bag5} follows from \eqref{assum1}. This implies that
\begin{align}
\max_{l\geq 1}\bigg|\frac{S_l}{\beta_l d_l}\bigg|<\infty, \qquad a.s.
\end{align}
Finally, we have
\begin{align}
0&\leq \bigg|\frac{S_l}{d_l b_l}\bigg|\\
&=\bigg|\frac{S_l}{d_l \beta_l}\bigg|\bigg|\frac{\beta_l}{b_l}\bigg|\\
&\leq \big(\sup_{l\geq 1}\bigg|\frac{S_l}{d_l \beta_l}\bigg|\bigg)\bigg|\frac{\beta_l}{b_l}\bigg|\\
& \to 0 \qquad a.s.
\end{align} as $l \to \infty$. 
This concludes the proof of Lemma \ref{GSLLN}.




\end{document}